\documentclass[10pt,journal,twoside,web]{ieeecolor}
\usepackage{generic}
\usepackage{cite}
\usepackage{amsmath,amssymb,amsfonts}
\usepackage{algorithmic}
\usepackage{graphicx}
\usepackage{textcomp}
\usepackage{hyperref}

\usepackage{dsfont}
\usepackage{graphicx}
\usepackage[ruled]{algorithm2e}
\newcommand{\gmin}{g_\textrm{min}(a^*)}

\DeclareMathOperator*{\argmax}{arg\,max}
\newtheorem{theorem}{Theorem}[section]
\newtheorem{lemma}[theorem]{Lemma}
\newtheorem{proposition}[theorem]{Proposition}

\newtheorem{definition}{Definition}

\def\BibTeX{{\rm B\kern-.05em{\sc i\kern-.025em b}\kern-.08em
    T\kern-.1667em\lower.7ex\hbox{E}\kern-.125emX}}
\begin{document}
\title{Achieving Robust Performance using Minimal Communication in Resource Allocation Games}
\author{Brandon Collins, \IEEEmembership{Member, IEEE}, Colton Hill, \IEEEmembership{Member, IEEE}, Philip N. Brown, \IEEEmembership{Member, IEEE} 
\thanks{This paragraph of the first footnote will contain the date on 
which you submitted your paper for review. Research was sponsored by the Air Force Office of Scientific Research under award number FA9550-23-1-0171 and by the Army Research Office under grant number W911NF-25-1-0239. The views and conclusions contained in this document are those of the authors and should not be interpreted as representing the official policies, either expressed or implied, of the AFOSR, Army Research Office, or the U.S. Government. The U.S. Government is authorized to reproduce and distribute reprints for Government purposes notwithstanding any copyright notation herein.}
\thanks{Brandon Collins is with the Thayer Department of Engineering, Dartmouth College, Hanover, 
NH 03755 USA (e-mail: Brandon.C.Collins@dartmouth.edu).}
\thanks{Colton Hill is with the Computer Science Department, University of Colorado Colorado Springs, Colorado Springs, CO 80918 USA (e-mail: chill13@uccs.edu).}
\thanks{Philip N. Brown is with the Computer Science Department, University of Colorado Colorado Springs, Colorado Springs, CO 80918 USA (e-mail: pbrown2@uccs.edu).}
}

\maketitle

\begin{abstract}

Increasingly, resource allocation games are being proposed to model team coordination in denied communication environments.
Of particular interest is understanding the impact of communication denial on the quality of emergent team coordination.
In this work, we consider the situation with an arbitrary communication network and use the Price of Anarchy to quantify the quality of emergent behavior. 
Our main result is a computationally efficient algorithm that calculates a minimal communication network that has the same performance guarantee as the full information case.
Additionally, we provide efficient algorithms that compute a Nash and a strict Nash equilibrium in the full information setting.
Finally to support these algorithms, we provide sufficient conditions for the existence of a strict Nash equilibrium, characterize strict Nash equilibria across all communication networks, and show that each strict Nash equilibrium in the full information case has a necessary and sufficient set of communication links that induce it.
We conclude by giving an execution time experiment of the proposed algorithm and examine several example games.
\end{abstract}


\section{Introduction}



The widespread proliferation of autonomous agents and teams has amplified the need to understand interaction and coordination between autonomous agents.
Of particular interest is scenarios where agents may have limited or denied communication networks.
One such example of this is diaster response, where teams of autonomous aerial vehicles have been proposed for a variety of roles such as search and rescue \cite{kuntze2012seneka,steinhausler2022detection,zhang2022training,tusnio2022efficiency,arnold2018search}, communication relay \cite{eom2019uav}, wildfire monitoring \cite{alexis2009coordination,tzoumas2023wildfire,pham2018distributed}, and fire suppression \cite{innocente2019self,alsammak2023nature,john2025resource}.
However, communication during disaster situations is often challenging due to damaged infrastructure \cite{comfort2006communication}, limited bandwidth \cite{manzini2023wireless}, and harsh terrain \cite{mcrae2021utilizing}.
Therefore, it is critical to characterize effective coordination of autonomous teams with limited communication.


Although a wide variety of algorithms and approaches have been studied for this problem \cite{nayak2020experimental,schneider2015auction,otte2020auctions,sujit2012multi,rachmut2023asynchronous,su2016coordination}, in this work we consider a game-theoretic perspective.
From this point of view, the \textit{resource allocation game} is a natural and well-studied model.
In a resource allocation game, there is a set of resources (here we regard a resource as a task that a single agent is capable of completing)
and a set of agents such that each agent must decide on what resource to cover (e.g., which task to complete).
As agents select resources to cover, the agents will behave \textit{strategically}, optimizing a local utility function that depends on the collective selection of all agents.
Using this setup, a planner who must make offline decisions such as what kind of vehicles to purchase, how many to purchase, and where to station them can use the resource allocation game to understand the \textit{emergent performance} of various configurations.


We call the situation where all agents are mutually aware of each other's actions the \textit{full communication} setting.
With full communication, the problem reduces to the classically posed resource allocation game \cite{rosenthal1973class,milchtaich1996congestion}, and the emergent behavior has been extensively studied \cite{ieong2005fast,ashlagi2006resource,ashlagi2007learning}.
In contrast, in the \textit{limited communication} setting, a given agent may only receive communication from a subset of all agents.
To specify this relationship, we endow the game with a \textit{communication network}, such that a directed link from agent $i$ to agent $j$ indicates that agent $i$ communicates its current action to agent $j$.
In this way, if there is no link present between the agents, then the agent must make decisions without knowledge of the other agent's action.
Intuitively, this can cause inefficient allocations where multiple agents who are not aware of each other cover the same resource.

One relevant line of work studies the problem as distributed optimization with sequential decision updates.
Although not posed as a denied communication setting, this can be thought of as a \textit{Directed Acyclic Graph} (DAG) communication graph, where an agent only has knowledge of agents that moved before it.
These works provide performance bounds in terms of the graph-theoretic properties of the DAG \cite{gharesifard2017distributed,grimsman2018impact} and the agent's decision structure  such as parallel execution \cite{sun2020distributed}, minimal algorithmic time complexity \cite{konda2022execution}, and augmented greedy policies \cite{grimsman2022impact}.

In the cases involving strategic interaction (i.e. agents may observe each other's actions and mutually respond), it becomes necessary to define performance with respect to \textit{Nash Equilibria} (NE) or the emergent behavior of the game.
One popular metric for measuring the relative quality of the emergent behavior  is the \textit{Price of Anarchy} (PoA) \cite{koutsoupias1999worst}, which measures a ratio of the worst performing equilibrium compared to the optimal joint behavior of agents.
Using PoA, a second line of work considers the situation where agents can either see other agents or are completely blind.
The authors of~\cite{grimsman2020cost} prove a novel lower bound on PoA across classes of games that contain $k$ blind agents.
This is further refined in \cite{seaton2022all} where worst case equilibria are shown to have inherent fragility.

Recent work has begun to consider the case with arbitrary communication networks.
For example, \cite{grimsman2022valid} examines PoA in the broader class of valid utility games and \cite{brown2023robust} considers the situation where agents may redesign their utility function when faced with denied communication.
Finally, linear programs are provided in \cite{singh2025optimal} that can compute the PoA for any communication network, and to optimize the utility function.
These results are used in \cite{singh2025worst} to show that a similar fragility result of worst case equilibrium extends to arbitrary information networks.

However, the structure and properties of the Nash equilibria induced by arbitrary communication networks is not well understood.
To address this, we focus on the issue of finding a minimal communication network such that there is no performance degradation from the full communication case.
In support of this goal, we also develop supporting theory which provides insights into the structure of NE across the spectrum of communication networks.
This provides several insights for a system planner, such as what links are most important to reinforce or the likelihood that a realized communication network will provide robust performance.
Our results can be summarized as follows:
\begin{enumerate}
    \item First, we formalize basic properties of NE across communication networks (Theorems~\ref{thm:sne exists} and \ref{thm:informed ne}).
    \item Second, we give a procedure (Algorithm~\ref{algo:min communication}) that computes a game with a minimal communication network while not degrading PoA.
    \item Third, we develop and verify procedures to compute NE efficiently (Algorithms~\ref{algo:NE} and \ref{algo:SNE}).
    \item Finally, we experimentally verify the time-complexity of the proposed method and give some illustrative examples.
\end{enumerate}


\section{Model}
\label{sec:model}
In this work, we consider a single-select resource allocation game in a communication denied setting.
We model $N=\{1,2,\dots,n\}$ agents that attempt to cover a set of resources $R=\{1,2,3,\dots,m\}$.
We assume that each resource $r$ has a value $v_r$, whose vector is given by $V=(v_1,v_2,\dots,v_m)$ where $v_i > 0$ for all $r\in R$.
We say that the \textit{action set} or the set of resources that agent $i$ can cover is written as $A_i\subseteq R$, and we write the collective set of actions by all agents as $A=A_1\times A_2\times \dots A_n$.
We use the term \textit{single-select} to mean that any action $a_i$ covers exactly one resource (that is, $a_i\in R$).
We denote a specific selection of actions as $a\in A$, and we write $a$ as $a=(a_i,a_{-i})$, where $a_i\in A_i$ is a specific coverage of agent $i$, and $a_{-i}\in A_{-i}$ is the collective actions of all agents $j\in N\setminus\{i\}$.

We denote the count of the number of agents that cover a resource $r\in R$ as $|a|_r=\sum_{i\in N}  \mathds{1}(r = a_i)$, where $\mathds{1}(\cdot)$ denotes the indicator function, and we denote the subset of resources being covered in action profile $a$ as $R(a)=\{r\in R : |a|_r\geq 1\}$.
Next, we define the set coverage basis function as follows:
\begin{equation} \label{eq:basis}
    b(|a|_r)=\begin{cases}
        1 & |a|_r \geq 1 \\
        0 & |a|_r = 0,\\
    \end{cases}
\end{equation}
which simply indicates if each resource is being covered by at least one agent.
Building on this, we define a \textit{welfare function}, which measures the aggregate performance of the agents:
\begin{equation} \label{eq:welfare}
    W(a) = \sum_{r\in R} v_r b(|a|_r).
\end{equation}
The welfare function simply sums the value of each resource if that resource was covered by at least one agent.
This is the \textit{objective function} of the system planner, the entity who seeks to optimize the aggregate performance of the agents.
However, each agent makes decisions independently based on available information. 
In addition to knowing the values of the resources $V$, the agents also receive some information about the actions of their peers.

We model a communication denied environment as each agent being able to receive communications from a subset of all agents $N$.
We denote the set of agents that agent $i$ can receive communications from as $\mathcal{N}_i\subseteq N$, and we say that communication from $j\in N$ to $i\in N$ was \textit{denied} if $j\notin \mathcal{N}_i$.
We always assume that an agent can see itself, that is, $i\in \mathcal{N}_i$, and we specify all communication links as $\mathcal{N}=(\mathcal{N}_1,\mathcal{N}_2,\dots,\mathcal{N}_n)$.
We term the special case where $\mathcal{N}_i=N$ for all $i\in N$ as the \textit{complete information} case, where all agents can receive communication from all other agents.
For a given collective choice of actions $a$, we compactly notate which actions can be seen by agent $i$ as
\begin{equation} \label{eq:info j to i}
    a_{j\rightarrow i}=\begin{cases}
        a_j & j\in \mathcal{N}_i \\
        \emptyset & j \notin \mathcal{N}_i
    \end{cases},
\end{equation}
 which we use to define the knowledge agent $i$ has about other agents as $\tilde{a}_{-i}=(a_{1\rightarrow i},a_{2\rightarrow i},\dots,a_{i-1\rightarrow i},a_{i+1\rightarrow i}\dots,a_{n\rightarrow i})$.

Using this information structure, we use the \textit{marginal contribution} utility function for each agent $i$:
\begin{equation} \label{eq:utiltiy}
\begin{aligned}
    U_i(a) &= W(a_i,\tilde{a}_{-i})-W(\tilde{a}_{-i}) \\
    &=\sum_{r\in R} v_r\mathds{1}(|a_i|_r-|\tilde{a}_{-i}|_r=1).
\end{aligned}
\end{equation}
This is the local utility for each agent $i$ and the maximization of this function is the only goal of each agent.
We can specify a full instance of a \textit{Communication Denied Resource Allocation Game} as $g=(N,A,R,V,\mathcal{N},U)$.
In this work, we consider the class of games $\mathcal{G}$ that are single-select ($A_i\subseteq R$), use the set cover basis function \eqref{eq:basis}, and marginal cost utility \eqref{eq:utiltiy}, while permitting the rest of the parameters ($N,A,R,V,\mathcal{N}$) to be arbitrary.
We denote the space of all such games as $G$, and $\bar{G}$ as the subset of games with full information (that is, $\mathcal{N}_i=N$ for all $i\in N$).
We refer to the subclass of games that have the same structure $(N,A,R,V)$ as a full information game $\bar{g}\in \bar{G}$ but with any communication network $\mathcal{N}$ as $G(\bar{g})$.

To characterize the emergent behavior of a game, we use the notion of NE.
A state $a^*\in A$ is an NE if for each agent $i\in N$,
\begin{equation} \label{eq:NE}
    U_i(a^*_i,a^*_{-i})\geq U_i(a'_i,a^*_{-i})
\end{equation}
for any unilateral action deviation $a'_i\in A_i$.
Similarly, in the case where the inequality is strict, $a^*$ is a \textit{Strict Nash Equilibrium} (SNE) if
\begin{equation}\label{eq:SNE}
    U_i(a^*_i,a^*_{-i})> U_i(a'_i,a^*_{-i})
\end{equation}
for any unilateral action $a'_i\in A_i$ for all agents $i$.
We let NE$(g)$ and SNE$(g)$ be the set of all Nash equilibria and strict Nash equilibrium for any game $g\in G$. 
In this setting, NE need not exist in a game with any given communication structure, and we provide sufficient conditions for the existence of Nash in Section~\ref{sec:sne algo}.
In the same section, we also provide sufficient conditions for the existence of SNE in the full information setting.

Nash equilibrium are often thought of as the long run behavior of a game after some learning process where agents continuously strategically react to each other's actions.
However, a game may have many equilibria and it is often difficult to predict which equilibrium (or set of equilibrium) agents may settle on.
Regardless of this complexity, it is desirable to have metrics that measure the quality of emergent behavior of a game.

One way to measure the quality of emergent behavior is the Price of Anarchy.
The PoA measures the welfare of the worst equilibrium compared to the action profile $a^*$ that maximizes the welfare function.
Formally, for a single game instance, this can be described as
\begin{equation}
    \textrm{PoA}(g) =\frac{\min_{a\in \textrm{NE}(g)}W(a)}{\max_{a\in A}W(a)}.
\end{equation}
The PoA is often interpreted as the worst case cost of distributed decision making.
In this work, we focus on SNE which gives a slight refinement of classical PoA:
\begin{equation}
    \textrm{sPoA}(g) =\frac{\min_{a\in \textrm{SNE}(g)}W(a)}{\max_{a\in A}W(a)}.
\end{equation}
Because $\textrm{SNE}(g)\subseteq \textrm{NE}(g)$, it must be $\textrm{sPoA}(g)\geq \textrm{PoA}(g)$, giving an improvement in some games.

\section{Computing a Minimal Communication Game}
\label{sec:min}

Suppose a system planner considers the full communication game $\bar{g}\in \bar{G}$ such that $\bar{g}$ is an arbitrary single-select resource allocation game, with marginal contribution utility functions and the set cover welfare function.
The system designer seeks to understand the performance of a communication denied game $g$, where $g\in G(\bar{g})$ has the same structure as $\bar{g}$ but with denied communication.
To address this, we develop an algorithm that computes the minimal communication graph such that $\textrm{NE}(g)=\{a^*\}$ for some $a^*\in \textrm{SNE}(\bar{g})$, which is sufficient to ensure $\textrm{sPoA}(g)\geq \textrm{sPoa}(\bar{g})$.
This approach can be thought of as computing the minimal amount of communication required such that agents exhibit similar strategic behavior to the full communication case.

Before computing the minimal communication games, we give two assumptions which we will use throughout the work.
Although individual results may only need one of the assumptions, the pipeline of algorithms presented will ultimately require both, and thus we are primarily concerned with games that satisfy both of them.
\begin{itemize}
    \item \textbf{Assumption A1} (\textit{non-triviality of agents}): There exists an action profile $a\in A$ such that $|a|_r\leq1$ for all $r\in R$,
    \item \textbf{Assumption A2} (\textit{distinct resources}): For any unique $a_i,a'_i\in A_i$ it must be $v_{a_i}\neq v_{{a_i'}}$.
\end{itemize}
Intuitively, Assumption~A1 ensures that there exists at least one action profile $a\in A$ such that no agents are choosing redundant actions.
If this assumption were not true, then in every $a\in A$ there would be at least one agent $i$ whose choice would have no impact on the welfare function, that is $W(a)=W(a_{-i})$.
In this case, the game $g$ could be reduced by removing such an agent $i$ iteratively until A1 is satisfied.
Assumption A2 simply requires that all agents have a strict preference over their available resources in the case when no other agent is covering any of them.

In the first, we show that Assumptions A1, A2 are sufficient to ensure that SNE exists in full communication games.
\begin{theorem}\label{thm:sne exists}
    Let game $\bar{g}\in \bar{G}$ have full communication and satisfy Assumptions A1, A2. 
    Then $|\textrm{SNE}(\bar{g})|\geq1$.
\end{theorem}
The proof of Theorem~\ref{thm:sne exists} appears in the appendix.
From a technical point of view, the immediate outcome of this theorem is that $\textrm{sPoA}(\bar{g})$ is well defined, since there must be at least one strict Nash equilibrium.

With the knowledge that SNE always exist, we now investigate their properties, beginning with a definition:
\begin{definition}[Informed Action]
    An action profile $a\in A$ is \textit{informed} if $|a|_r\leq 1$ for all $r\in R$.
\end{definition}
Intuitively, this definition describes the class of all action profiles such that no agent double covers any resource.
We use the term \textit{informed} in the sense that agents appear to be aware of the actions of other agents (in a given communication network this may not be nominally true) and are using that information to strategically make choices that improve the overall welfare function.
In the following theorem, we connect informed actions with the strict Nash equilibrium of the full information game (SNE$(\bar{g})$).
\begin{theorem} \label{thm:informed ne}
        Let full communication game $\bar{g}\in \bar{G}$ satisfy Assumptions A1 and A2,  then $a^*\in \textrm{SNE}(\bar{g})$ if and only if $a^*\in\cup_{g\in G(\bar{g})} \textrm{NE}(g)$ and $a^*$ is informed.
\end{theorem}
The proof of Theorem~\ref{thm:informed ne} appears in the appendix.

This theorem yields a few insights into communication denied settings.
First in the forward direction, it characterizes that all of the SNE of a full information game $\bar{g}$ are informed.
That is, if an equilibrium $a^*\in\textrm{NE}(g)$ in a communication denied game $g$ is uninformed, that inefficiency can be directly attributed to the communication denial.
In the backward direction, it gives that SNE$(\bar{g})$ is the complete set of all informed equilibria across all communication profiles.
Specifically, it is not possible to construct a communication network that induces an informed Nash equilibrium that is not part of the SNE in the fully informed setting.

\begin{algorithm}
\label{algo:min communication}
\caption{Compute a Minimal Communication Game}
\KwData{Full information game $\bar{g}\in G$ that satisfies Assumptions A1 and A2.}
\KwResult{$g_{\min}(a^*)\in G(\bar{g})$ with minimal communication with respect to some $a^*\in \textrm{SNE}(\bar{g})$}
$a^*\in \textrm{SNE}(\bar{g})$\;
\For {$i\in N$}{
    $\mathcal{N}_i\gets \emptyset$\;
    \For {$j \in N$}{
        \If{$a^*_j\in A_i$ and $v_{a^*_j}>v_{a^*_i}$}{
            $\mathcal{N}_i\gets \mathcal{N}_i\cup \{j\}$
        }
    }
}
\Return $g_{\min}(a^*)\in G(\bar{g})$ with communication structure $\mathcal{N}$\;
\end{algorithm}

\begin{figure}
    \centering
    \includegraphics[scale=0.25]{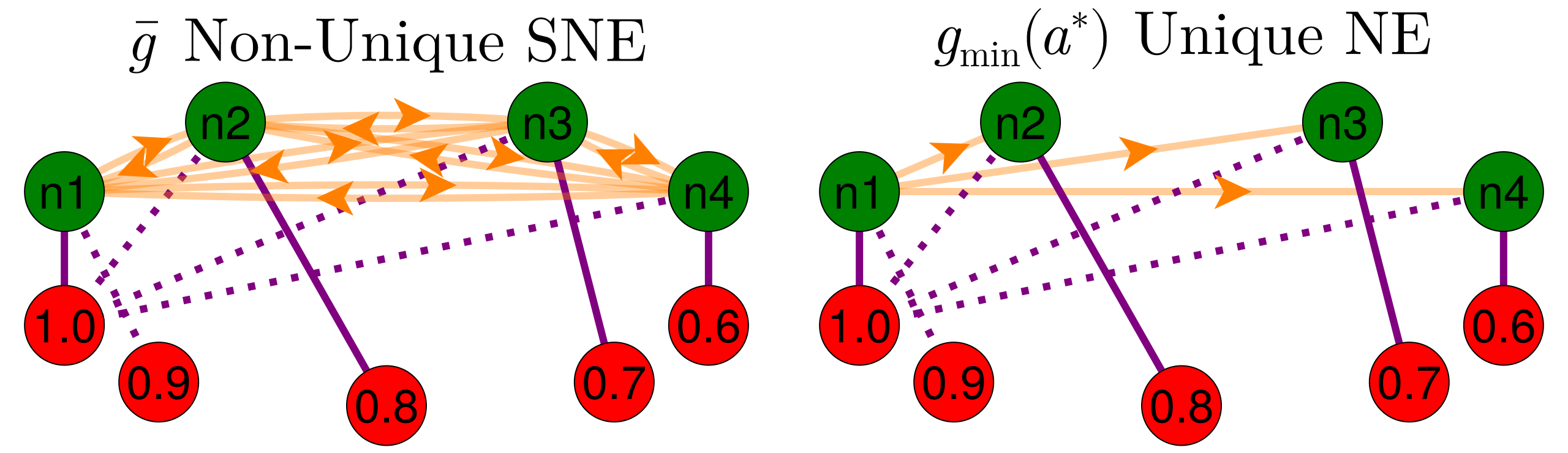}
    \caption{
    Examples Of Algorithm~\ref{algo:SNE} and Algorithm~\ref{algo:min communication}.
    The green circles indicate the agents, $N=\{n1,n2,n3,n4\}$ and the red circles indicate the resources, where the number in the middle is the value $v_r$.
    The orange lines show communication links in the direction information flows, formally an arrow from agent $i$ to $j$ indicates $i\in \mathcal{N}_{j}$.
    An agent's action set is indicated in purple, where the dotted arrows an agent's available but not taken actions, and the solid line indicates a particular selection $a_i\in A_i$.
    The left diagram shows an example strict Nash equilibrium $a^*$, shown as the solid purple lines, computed via Algorithm~\ref{algo:SNE} (taking $\bar{g}$ as input) and the right diagram shows $\gmin$, output of Algorithm~\ref{algo:min communication} (taking $\bar{g}$ and $a^*$ as input).
    }
    \label{fig:star}
\end{figure}

Using these results, Algorithm~\ref{algo:min communication} gives the procedure to compute a minimum communication game, and an example is given on the right side of Figure~\ref{fig:star}.
The algorithm iterates over the agents in a double nested loop and adds a communication link $j\in \mathcal{N}_i$ whenever $a^*_j\in A_i$ and $v_{a^*_j}>v_{a^*_j}$.
Intuitively, this ensures that any time an agent would desire to switch their action from $a^*_i$ to some $a_i$ (because $v_{a^*_i}>v_{a^*_i}$), the agent is informed that some agent $j$ is covering it as $a^*_j=a_i$, and therefore $U_i(a_i,a^*_{-i})=0$.
This structure also implies that the time complexity of Algorithm~\ref{algo:min communication} is $\mathcal{O}(n^2)$.
We call the computed game a \textit{minimal communication game} and denote it by $\gmin$.
Although seemingly a simplistic process for constructing a communication structure, such a game has several interesting properties:
\begin{theorem}
\label{thm:gmin}
    Let $\bar{g}\in \bar{G}$ be a full information game that satisfies Assumptions A1 and A2, and let $\gmin\in G(\bar{g})$ be the result of Algorithm~\ref{algo:min communication} with respect to $a^*\in \textrm{SNE}(\bar{g})$.
    Letting $\mathcal{N}$ denote $\gmin$'s communication structure, the following statements are all true:
    \begin{enumerate}
        \item $\textrm{NE}(\gmin)=\{a^*\}$,
        \item $\textrm{sPoA}(\gmin)\geq \textrm{sPoA}(\bar{g})$ and $\textrm{PoA}(\gmin)\geq \textrm{PoA}(\bar{g})$,
        \item if game $g'\in G(\bar{g})$ has $\mathcal{N}_i\nsubseteq \mathcal{N}'_i$ for some $i \in N$ then $a^*\notin \textrm{SNE}(g')$, and
        \item if game $g''\in G(\bar{g})$ has $\mathcal{N}_i\subseteq \mathcal{N}''_i$ for all $i \in N$ then $a^*\in \textrm{SNE}(g'')$.
    \end{enumerate}
\end{theorem}
The proof of Theorem~\ref{thm:gmin} appears in the appendix.

In summary, the first result gives that $a^*$ is the unique Nash equilibrium of $\gmin$.
The second result of the theorem directly gives the desired performance result.
Therefore, Algorithm~\ref{algo:min communication} can be viewed as an efficient algorithm that computes a minimal communication structure such that the emergent behavior of the game does not deteriorate from the full information setting.
This may be useful in a variety of circumstances, such as the setting where a designer can plan the communication network, decide which links to reinforce or repair, or rapidly evaluate outcomes in a wide range of potential emergent communication networks.

The third result is the reason for calling $\gmin$  \textit{minimal}.
That is, if some game $g'$ is missing even a single communication link from $\gmin$, then $a^*$ is not an equilibrium in $g'$.
Intuitively, every communication link of $\gmin$ is \textit{necessary} for the equilibrium $a^*$ to be present in any game, regardless of the rest of the communication structure. 
The fourth result gives that if a game $g''\in G(\bar{g})$ has strictly more communication, (i.e. $g'\in G(\bar{g})$ such that  $\mathcal{N}_i\subseteq \mathcal{N}'_i\ \forall i\in N$), then $a^*$ must be a strict Nash equilibrium in that game.
Equivalently, the communication links of $\gmin$ are \textit{sufficient} to ensure that any game with at least that subset of communication will have equilibrium $a^*$.
Together, the third and fourth results show that the set of communication links in $\gmin$ are necessary and sufficient such that $a^*$ is a strict Nash equilibrium in any game $g\in G(\bar{g})$.

We use the necessity and sufficiency of the third and fourth results to partition the space of communication networks in $G(\bar{g})$.
For this discussion, let the set of all minimal communication games for a full communication game $\bar{g}$ be denoted by $G_{\min}(\bar{g})=\{\gmin: a^*\in \textrm{SNE}(\bar{g})\}$.
We call game $g\in G(\bar{g})$ \textit{high-information} if there exists $g'\in G_{\min}(\bar{g})$ such that $\mathcal{N}_i\subseteq \mathcal{N}'_i\ \forall i\in N$.
Notably, this is a sufficient condition for the existence of NE which is convenient for NE based analysis approaches such as PoA.
Conversely, if no such $g'\in G_{\min}(\bar{g})$ exists then we say the game is \textit{low-information}.
Game $g$ need not have an equilibrium, but if it does it must be uninformed via Theorems~\ref{thm:informed ne} and \ref{thm:gmin}.

\section{Computing a Nash Equilibrium}
\label{sec:ne}

As Algorithm~\ref{algo:min communication} requires an SNE as input, we develop a general procedure for computing SNE efficiently.
This procedure takes two steps: first we find a Nash equilibrium and then we use a best response search process to find a SNE based on a NE.
In this section, we give the provide the algorithm to efficiently compute an NE, and in the next section we give the procedure to find an SNE.

Because each minimum communication structure is associated with a specific SNE and in this approach we compute SNE's from an NE, we begin by designing an algorithm to compute NE quickly.
In general, it is well understood that in potential games, Nash equilibria always exist \cite{monderer1996potential} and asynchronous best response processes converge to them \cite{marden2012revisiting}.
All full communication games $\bar{g}\in \bar{G}$ are potential games because the marginal contribution utilities \eqref{eq:utiltiy} enables the welfare function \eqref{eq:welfare} to double as a potential function.
Additionally, the topic of learning a Nash equilibrium in (full information) single-select resource allocation game has been studied in \cite{ieong2005fast}, which provides lower and upper bounds of the number of learning updates to learn a Nash equilibrium.
Compared to our model, \cite{ieong2005fast} makes the assumption that $A_i=R$ for all agents $i$, which is termed a \textit{symmetric} game.
Particularly, it provides a lower bound of $\Omega(\min(n^2,nm))$ as the fewest number of better response updates such that an equilibrium can be learned by agents.

We design Algorithm~\ref{algo:NE} as a two iteration asynchronous strict best response process with time complexity $\mathcal{O}(n|A_*|)$, where we denote the size of the largest action set of any agent as $|A_*|=\max_{i\in N}(|A_i|))$.
Notably, our algorithm recovers complexity $\mathcal{O}(nm)$ in the case where $|A_*|=m$, and offers improved complexity in the case where agents have smaller action sets.
Although~\cite{ieong2005fast} does not provide bounds for the non-symmetric case, we highlight that our approach is known to be optimal in the case of symmetric single-select games.
We consider the augmented action space $\tilde{A}_i=A_i\cup\{0_i\}$ where virtual resource $0_i$ is a private resource that only $i$ can cover that results in negative welfare $W(0_i)=-\infty$.
Note, $U_i(a_i,a_{-i})>U_i(0_i,a_{-i})$ for all $a_i\in A_i$ and $a_{-i}\in A_{-i}$, so all agents will best respond off of this action in the first iteration of Algorithm~\ref{algo:NE}, giving $a^1,a^2\in A$.

\begin{algorithm}
\caption{Find a Nash Equilibrium}
\label{algo:NE}
\KwData{Full information game $\bar{g}\in G$ that satisfies Assumption A2}
\KwResult{$a^*\in\textrm{NE}(\bar{g})$}
$a^0\gets (0_1,0_2,\dots,0_n)$\;
\For{$t\in \{1,2\}$}{
    $a^t\gets a^{t-1}$\;
    \For{$i\in N$}{
        $a'_i\in \argmax_{a_i\in \tilde{A}_i}U_i(a_i,a^t_{-i})$\;
        \If{$U_i(a'_i,a^t_{-i})>U_i(a^t_i,a^t_{-i})$}{
        $a^t_i\gets a'_i$\;
        }
    }
}
\Return $a^2$
\end{algorithm}
We now verify that Algorithm~\ref{algo:NE} returns a Nash equilibrium.
We begin by giving a lemma that gives a useful relationship between $a^1$ and $a^2$ of this algorithm.

\begin{lemma}\label{thm:monotinicity}
    If $a^1,a^2\in A$ are the action profiles resulting from Algorithm~\ref{algo:NE} on $\bar{g}\in \bar{G}$ such that $\bar{g}$ satisfies Assumption A2 then $R(a^1)\subseteq R(a^2)$.
\end{lemma}
The proof of Lemma~\ref{thm:monotinicity} appears in the appendix. 
The above lemma gives some structure on the emergent action profile as the best response process in Algorithm~\ref{algo:NE} proceeds.
Particularly, if an agent is the only one covering a resource, it will never switch off that resource.
Using this lemma, we give Proposition~\ref{thm:find Nash}, which verifies Algorithm~\ref{algo:NE} indeed returns a Nash equilibrium.

\begin{proposition} \label{thm:find Nash}
    Let $\bar{g}\in G$ be a full information game that satisfies assumption A2, then running Algorithm 1 on $\bar{g}$ returns $a^*\in \textrm{NE}(\bar{g})$.
\end{proposition}
The proof of Proposition~\ref{thm:monotinicity} appears in the appendix.
This Theorem provides the ability to quickly find NE for the next steps of the algorithm.
Notably, the algorithm can also be adapted to find multiple NE by changing the order of best responses.

\section{Computing A Strict Nash Equilibrium}
\label{sec:sne algo}

Using the results characterizing the structure and existence and SNE, we seek to efficiently compute them.
We will begin with an NE (computed via Algorithm~\ref{algo:NE}) and search for a sequence of weak best responses such that the agents cover a superset of resources which will lead to an SNE.
To do this, we will reduce the best response problem to a graph theoretic shortest path problem, where each edge in the graph will correspond to a best response by an agent.
The main advantage of this approach is that the search for a best response sequence can be handled by well studied and efficient graph theoretic algorithms, and we only need to show that such a path exists.
We begin by providing a restricted definition of best response processes.
\begin{definition}[Weak Best Response Sequence]
A Weak Best Response Sequence is a series of action profiles of length $T\leq n$ given by  $(a^0,a^1,a^2,\dots, a^T)$ such that:
\begin{enumerate}
    \item (\textit{Asynchronous Updates}) For every pair of consecutive profiles $a^t,a^{t+1}$ there must exist a unique \textit{updating} agent $i$ such that $a^t_{-i}=a_{-i}^{t+1}$ and $a^t_{i}\neq a_{i}^{t+1}$.
    \item (\textit{Single Update per Agent}) Each agent only updates their action once. If agent $i$ updates their action at time $t$, then  $a^0_i=a^1_i=\dots=a^{t-1}_{i}$ and $a^t_i=a^{t+1}_i=\dots=a^{T}_{i}$.
    \item (\textit{Consecutive Updates}) If agent $i$ changes their action at time $t$ such that $a^{t-1}_i\neq a^t_i$, then an agent $j$ with $a^0_j=a^t_i$ must be the updating agent at $t+1$, that is $a^{t+1}_{j}\neq a^{t}_{j}$.
    If no such agent $j$ exists then $t=T$ and the sequence terminates. 
    \item (\textit{Weak Best Response}) Any agent $i$ that updates their action $a^{t}_i\neq a^{t+1}_i$ must weakly improve their utility $U_i(a^{t+1})\geq U_i(a^{t})$.
\end{enumerate}
\end{definition}

The above definition restricts potential weak best response sequences such that only one agent may update at a time (this property is often called \textit{asynchronous}), and an agent can only wake up for an update if at the previous timestep an agent took its current action.
Additionally, each agent is only permitted to update their action once, must change their selection from their current choice, and must make a choice that weakly improves their utility.
This definition restricts the number of possible sequences (and therefore the search space).



Using this definition, we seek to find the best response sequence that $(a^0\cdots a^T)$ such that $R(a^0)\subset R(a^T)$, meaning that the best response sequence causes the agents to find a more efficient allocation, formally $W(a^T)>W(a^0)$.
To do this, we partition the agents into the following three sets:
\begin{align}
    &S(a)=\{i\in N : |a|_{a_i}\geq 2\} \label{eq:source}\\
    &C(a)=\{i\in N : |a|_{a_i}= 1\textrm{ and } |a|_{r}\geq 1 \quad \forall r\in A_i\} \label{eq:connecting}\\
    &D(a)=\{i\in N : |a|_{a_i}= 1 \textrm{  and }\exists r\in A_i\textrm{ s.t. }|a|_r=0\} \label{eq:destination}
\end{align}
where if $a\in\textrm{NE}(\bar{g})$ and $\bar{g}$ satisfies assumptions A1,
then the above sets are disjoint and $S(a)\cup C(a)\cup D(a)=N$.
Going set by set, agents $s\in S(a^0)$ implies $U_s(a^0)=0$, and because $a^0\in \textrm{NE}(\bar{g})$, it must be $U_s(a_s,a^0_{-s})=0$ for all $a_s\in A_s$.
Therefore, we call them \textit{source} nodes in the sense that they are free to best respond to any $a_i\in A_i$ without any action by other agents.
The significance of this property is that they can begin a weak best response sequence.
Agents $c\in C(a^0)$ are called \textit{connecting} agents.
The key difference is that they require some agent $i$ to select $a_i=a^0_c$, and then similarly to $s\in S(a)$, they are free to best respond to any action $a_c \in A_c$.
Finally, nodes $d\in D(a^0)$ are termed \textit{destination} nodes.
Intuitively, if some agent $i$ selects $a_i=a^0_d$, then agent $d$ will be free to select action $a_d$, and cause a super set of actions to be covered, as desired.
Using these definitions, we now show that such best response processes always exists.
We now give the result that in a weak Nash equilibrium, weak best response sequences always exist.
\begin{theorem} \label{thm:SWBRP}
    Let game $\bar{g}\in \bar{G}$ has full information and satisfies Assumption A1. 
    If $a^0\in \textrm{NE}(\bar{g})$ but $|a^0|<n$ then there exists a weak best response sequence $(a^0\cdots a^T)$ such that $R(a^0)\subset R(a^T)$ and $a^T\in \textrm{NE}(\bar{g})$.
\end{theorem}
The proof of Theorem~\ref{thm:SWBRP} appears in the appendix.

Although in the remainder of this section we will utilize Theorem~\ref{thm:SWBRP} to compute a strict Nash equilibrium, in isolation it serves as an equilibrium selection argument in favor of SNE.
That is, from any weak Nash equilibrium there exists a sequence of best response actions such that $R(a^0)\subset R(a^T)$, improving the welfare.
Because the new action profile $a^T$ is also a Nash equilibrium, the theorem can be reapplied and the agents can repeat this process once again.
Eventually, $|R(a)|=n$ and by Theorem~\ref{thm:informed ne} it must be that $a$ is a strict Nash equilibrium, at which point no agent can use a best response to change their action.
Therefore, if agents update their action in any process that permits weak best responses, strict Nash equilibrium will be the natural stable emergent outcome.
This provides formal justification for using $\textrm{sPoA}$ as the primary equilibrium quality metric, in place of the traditional notion of PoA.


Now it is known that there always exists a weak best response sequence, all that remains is finding one in as few computations as possible.
Unfortunately, the proof of Theorem~\ref{thm:SWBRP} depends on knowing an informed action profile $a\in A$, which itself needs to be computed.
Instead, we will reformulate the problem to a graph theoretic shortest path problem.
Formally, we define a best response graph for a Nash equilibrium $a^0\in \textrm{NE}(g)$ as $\gamma(a^0)=(N,E)$ where the agents $N$ are the nodes and edges $(i,j)\in E$ exist whenever $a_j\in A_i$ and $i\notin D(a^0)$.
We say that a path on graph $\gamma(a)$, is a sequence of nodes $(i^1,i^2,\dots)$ such that every consecutive pair of nodes $(i^t,i^{t+1})$ satisfies $(i^t,i^{t+1})\in E$.
We interpret this definition as moving from node $i^t$ to $i^{t+1}$ on graph $\gamma(a)$ corresponds to agent $i^t$ updating their action to agent $i^{t+1}$'s original action $a^0_{i^{t+1}}$.
This creates new action profile $a^t=(a^0_{i^{t+1}},a^{t-1}_{-i^t})$.

We highlight that any path on $\gamma(a^0)$ can be mapped to a valid weak best response sequence $(a^0,a^1,a^2,\dots)$.
This follows because (i) each step in the graph only features one update, (ii) each node is accessed only once in a path meaning each node only updates once in the sequence, (iii) the edge structure ensures the consecutive update property is satisfied, and (iv) the weak best response is sastisfied as $i\in S(a^0)\cup C(a^0)$ are free to weakly best respond to any action, and the process terminates when a node $j\in D(a^0)$ updates.
This argument can be given similarly the other way that any weak best response sequence can be uniquely interpreted as a path on $\gamma(a^0)$.

\begin{algorithm}
\label{algo:SNE}
\caption{Find a Strict Nash Equilibrium}
\KwData{Full information game $\bar{g}\in G$ that satisfies Assumptions A1 and A2} 
\KwResult{$a^*\in\textrm{SNE}(\bar{g})$}
$a^*\in \textrm{NE}(\bar{g})$\;
\While{$|R(a^*)|<n$}{
    Compute $\gamma(a^*)$\;
    Select any $s\in S(a^*)$\;
    $(a^0,a^1\cdots a^T)\gets\textrm{BFS}(\gamma(a^*),s)$\;
    $a^*\gets a^T$\;
}
\Return $a^*$\;
\end{algorithm}

Now that we have reduced Weak Best Response Sequences to an equivalent graph representation, and Theorem~\ref{thm:SWBRP} gives that such a path must exist, all that remains is to deploy a graph-theoretic shortest path algorithm.
In this work, we use a single source, multi-destination variant of Breadth First Search (BFS) which is known to terminate in $\mathcal{O}(n+|E|)$ where $|E|\leq n|A_*|$, and thus must terminate in $\mathcal{O}(n|A_{*}|)$.
For convenience, we write the Breadth First Search Algorithm as $\textrm{BFS}(\gamma(a^0),s)=(a^0,a^1\dots a^T)$, where $s\in S(a^0)$ is the starting node.
For convenience, we define its return as a best response process, using the previous definition of $\gamma(a^0)$ to convert the path back into a weak best response sequence.

The process for computing a strict Nash equilibrium is given as Algorithm~\ref{algo:SNE}.
The algorithm begins with a full communication game $\bar{g}\in\bar{G}$ as input first utilizes Algorithm~\ref{algo:NE} to compute a Nash equilibrium $a^*$.
Then, the algorithm loops whenever $|R(a^*)|<n$ where this condition implies $a^*$ is uninformed and by Theorem~\ref{thm:informed ne} cannot be an SNE of $\bar{g}$.
In each loop, it computes $\gamma(a^*)$, which will terminate $\mathcal{O}(n|A_*|)$ as each agents action set needs to be iterated over to determine if there should be an edge or not.
Then, it selects $s\in S(a^*)$ arbitrarily and then executes $\textrm{BFS}(\gamma(a^*),s)$ and updates $a^*$ to be $a^T$ from BFS.
An example of a computed strict Nash equilibrium can be seen on the left side of Figure~\ref{fig:star}.

The key insight of this algorithm is that Theorem~\ref{thm:SWBRP} gives $a^T\in \textrm{NE}(\bar{g})$ which allows the algorithm to restart the process, each time improving $R(a^*)$.
Then, within $n$ iterations it must be that $|R(a^*)|=n$, which is a sufficient condition such that $a^*$ is informed, which by Theorem~\ref{thm:informed ne} gives $a^*\in \textrm{SNE}(\bar{g})$ as desired.
Further, because the algorithm pessimistically takes $n$ iterations, and both the inner computations $\gamma(a)$ and $\textrm{BFS}(\gamma(a),s)$ have time complexity $\mathcal{O}(n|A_*|)$, Algorithm~\ref{algo:SNE} has time complexity $\mathcal{O}(n^2|A_*|)$.
Note, if Algorithm~\ref{algo:min communication} is called using an SNE computed from Algorithm~\ref{algo:SNE} (which itself calls Algorithm~\ref{algo:NE}), then Algorithm~\ref{algo:SNE}'s time complexity will dominate, resulting in a final time complexity of $\mathcal{O}(n^2|A_*|)$ for the full algorithm stack presented in this paper.

As a final point, we highlight that our Algorithm~\ref{algo:SNE} is only one way to find SNE for Algorithm~\ref{algo:min communication}.
Recent work, \cite{kleer2017potential,Bilo2025Minimizing} provides efficient (in the single-select case) methods to compute the maximum of Rosenthal's potential function \cite{rosenthal1973class}, which is well known to be the potential function of congestion games (equivalently resource allocation games).
This is relevant because in this setting, the potential function can be reduced to the welfare function $W$, so the \textit{socially optimal} $a^*\in \argmax_{a\in A}W(a)$ can be efficiently computed.
Then, we have that $a^*\in \textrm{NE}(\bar{g})$ because $a^*$ is a maximizer of the potential function, which lets us begin to reason about $a^*$ using the results of this paper.
Particularly, Theorem~\ref{thm:SWBRP} gives it must be that $a^*$ is informed, or there would exist an Nash equilibrium $a$ such that $R(a^*)\subset R(a)$, which would imply $W(a)>W(a^*)$, contradicting optimality of $a^*$.
Because $a^*$ is an informed Nash equilibrium, it therefore must also be an SNE of $\bar{g}$ by Theorem~\ref{thm:informed ne}, which implies $a^*$ is a valid input of Algorithm~\ref{algo:min communication} and $\gmin$ is well defined.
This implies that $\textrm{PoA}(\gmin)=1$ as $a^*$ is its unique equilibrium.
Counter intuitively,  if a system designer who optimizes for PoA can control the communication network $\mathcal{N}$, they should prune the communication links in $\bar{g}$ until $\gmin$ is obtained.

In contrast to these approaches that find the optimal NE, we highlight that our Algorithms~\ref{algo:NE} and \ref{algo:SNE} are easy to adapt to search for many SNE (for example, by changing the behavior of the BFS search algorithm).
This is useful in denied communication settings, as in many applications such as natural disasters or enemy jamming the communication networked is realized in an online and unpredictable fashion.
In this way, our work can be used to inform a designer what a minimal set of efficient communication structures is.

\section{Examples and Simulations}
\label{sec:ex}

\begin{figure}[!ht]
    \centering
    \includegraphics[scale=0.25]{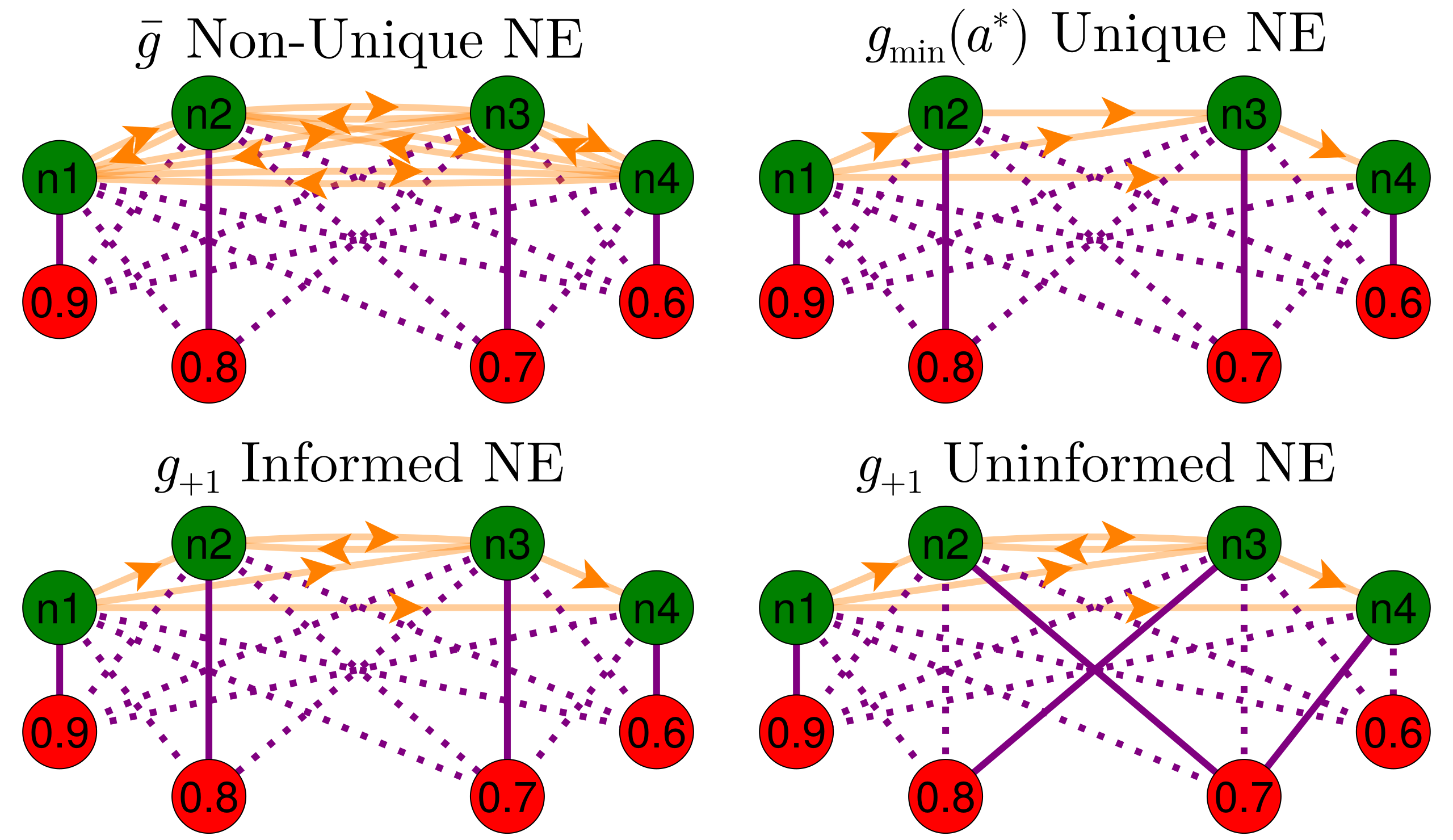}
    \caption{Examples of High-Information Games.
    Please see Figure~\ref{fig:star} for instructions on how to interpret each diagram as a resource allocation game.
    All action profiles showcased above are Nash equilibria.
    }
    \label{fig:addcomm}
\end{figure}
In this section, we present example games examining high-information games and low-information games.
Then, we conclude by giving an experiment which showcases Algorithm~\ref{algo:min communication} execution time on large games.
All figures and experiments presented in this work can be reproduced by source code available on \href{https://github.com/brandon-is-coding/Minimal_Communication_Games}{GitHub}.

Beginning with high-information games, Figure~\ref{fig:addcomm} depicts a variety of related high information games.
In the top left, the full information game $\bar{g}\in\bar{G}$ is depicted.
In this game all agents can cover all resources, with the exception that the resource worth 0.8 cannot be covered by agent $n4$.
The depicted action profile (in solid purple lines) represents the computed strict Nash equilibrium from Algorithm~\ref{algo:SNE}.

On the top right the associated minimal communication game $\gmin$ is depicted, as computed by Algorithm~\ref{algo:min communication}.
Here, an intuitive connection can be seen between the communication links in $\gmin$ and the equilibrium $a^*$ via the emergent DAG  communication structure (Algorithm~\ref{algo:min communication} will always produce a DAG).
Particularly, the agent who takes the highest value resource (in this case agent $n1$) can be thought of as moving first and without regard for any other agent, so they need no incoming communication links.
Then, the next agent who takes the second highest value resource (agent $n2$) receives information from the previous agents which effectively removes previously selected resources (the $0.9$ resource) from their action set.
This process continues, constructing a one way flow of information, such that agents who select lower value resources information that higher valued resources are already taken when necessary.
This sequential intuition of resource selection extends into the proof of Theorem~\ref{thm:gmin}, where it is formalized as mathematical induction.

Then, in the bottom two diagrams of Figure~\ref{fig:addcomm}, we construct a game $g_{+1}$ by adding a single communication link $n3\in \mathcal{N}_{n2}$ so that $n2$ now receives information from $n3$.
Notably, this new game has two equilibria which are depicted on the left and right.
On the left, $a^*\in \textrm{NE}(g)$ is depicted, which is required by Theorem~\ref{thm:gmin} (because $g_{+1}\geq_\mathcal{N}\gmin$).
However, on the right an uninformed equilibrium has been introduced because the 0.8 and 0.7 resources can now be covered by either $n2,n3$.
The issue is caused because agent $n4$ only receives information from $n3$ and implicitly expects $n3$ to cover the 0.7 resource.
In this way, minimal communication games are fragile in that $\textrm{sPoA}(\gmin)\geq \textrm{sPoA}(\bar{g})$, however $\textrm{sPoA}(g)<\textrm{sPoA}(\bar{g})$.
Formally, we can compute $\textrm{sPoA}(\bar{g})=\textrm{sPoA}(\gmin)=\frac{3}{3}=1$ but $\textrm{sPoA}(g)=\frac{2.4}{3}=0.8$.
This shows that even though the 4th statement of Theorem~\ref{thm:gmin} seems very promising that adding  communication links should improve the quality of emergent behavior, it is not strictly true with respect to $\textrm{sPoA}$ or PoA.

\begin{figure}
    \centering
    \includegraphics[scale=0.25]{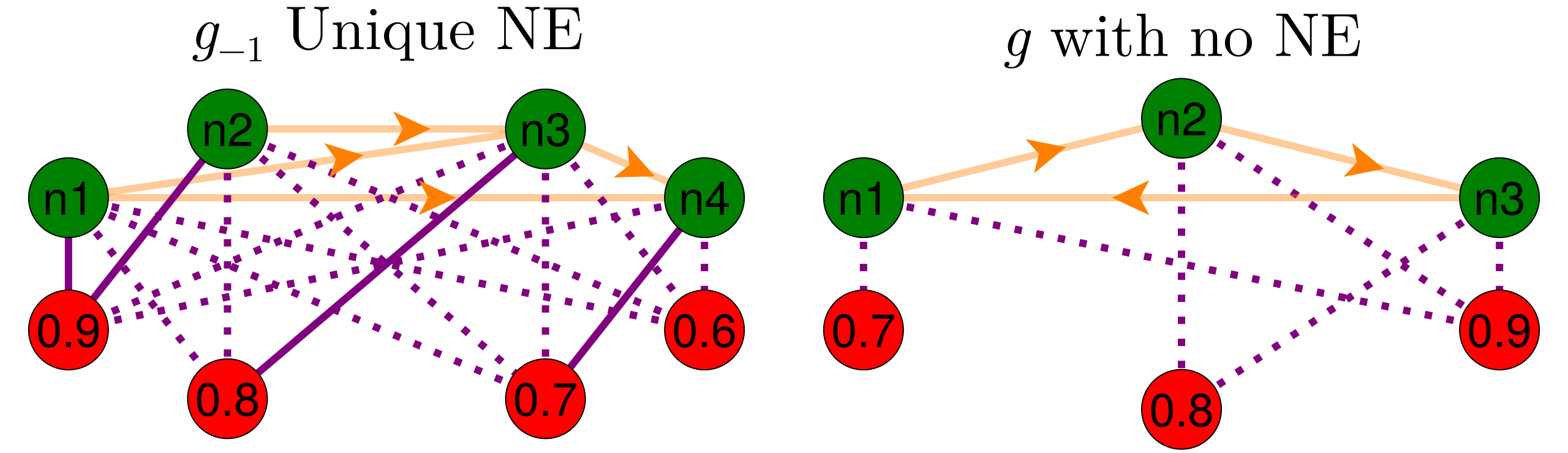}
    \caption{Examples of Low-Information Games.
    Please see Figure~\ref{fig:star} for instructions on how to interpret each diagram as a resource allocation game.
    On the left, game $g_{-1}$ is game $g_{\min}(a^*)$ from Figure~\ref{fig:addcomm} with the communication link from $n1$ to $n2$ removed.
    On the right, a game with no Nash equilibrium is showcased.
    \label{fig:lowinfo}
    }
\end{figure}

Next, in Figure~\ref{fig:lowinfo} we examine low-information games.
On the left, we depict the game $g_{-1}$, named such because it is the same as the previously discussed game $g_{\min}(a^*)$ with the communication link from $n1$ to $n2$ removed.
This game is a case study on how the  equilibrium quality degrades as games become low information.
Interestingly, the agents remain relatively strategic as only one pair of agents double cover resources.
In fact, agents $n3$, $n4$ act optimally given the circumstances.
In the same way, it can be seen that removing further links can cause agents to all pile onto high value resources.
For example, removing the link from $n1$ to $n4$ would cause $n4$ to cover the high value resource.

On the right, a different low-information game $g$ with no Nash equilibrium is shown.
It can be seen that the cycle of information causes unilaterial updates by agents to never settle into an equilibrium.
Supposing the agents update in order, agent $n1$ will select the 0.9 value resource, then with knowledge of that $n2$ will cover the 0.8 value resource.
However, agent $n3$ does not know $n1$ covers the 0.9 resource, so will cover it.
Then, agent $n1$ will see agent $n3$ has covered it, and switch to the 0.7 resource, causing $n2$ to now cover the 0.9 resource.
This cycle will repeat indefinitely, and no action profile will satisfy all agents.
Together, these two examples highlight the low-information game space, as described by Theorem~\ref{thm:gmin} and Theorem~\ref{thm:informed ne}.
First, it may be that there is no equilibrium present.
Second, if there is an equilibrium, then it is guaranteed that it is not informed.

\begin{figure}
    \centering
    \includegraphics[scale=0.08]{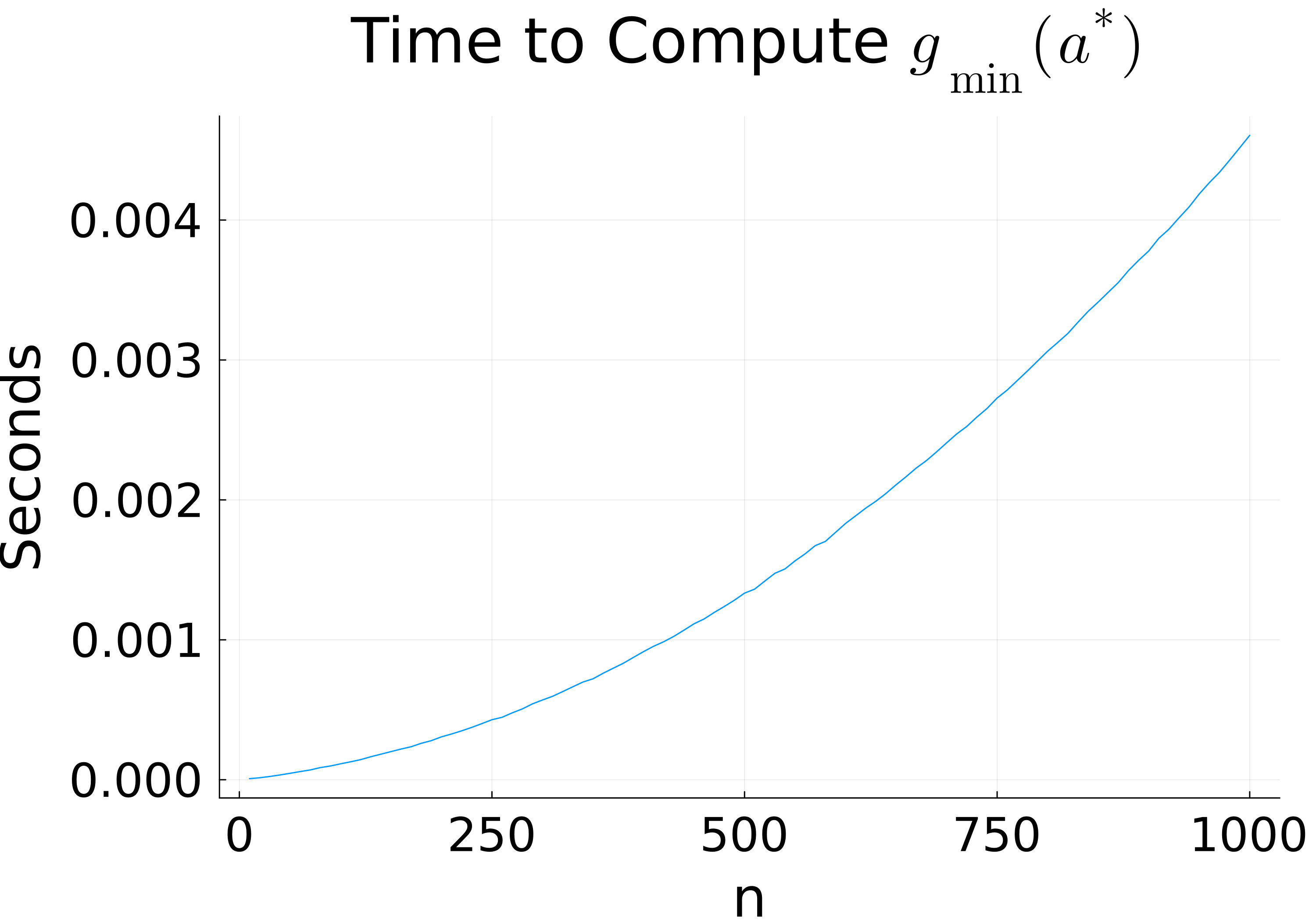}
    \caption{Time to Compute a Minimum Communication Game via Algorithm~\ref{algo:min communication}. Experiment conducted on an M4 iMac.}
    \label{fig:runtime}
\end{figure}

Finally, we showcase the runtime of Algorithm~\ref{algo:min communication} (which calls Algorithms~\ref{algo:NE} and \ref{algo:SNE}) in Figure~\ref{fig:runtime}.
For this experiment, we construct large resource allocation games.
Given values $n\in \{10,20,30,\dots,1000\}$, we create $n$ resources of slightly decrementing values $V=\{1-x\epsilon : x\in N\}$ for a small positive $\epsilon$.
Then, for each agent we randomly select 10 resources as the agent's coverage set $A_i$.
Finally, to ensure that Assumption A1 is satisfied, we create $n$ more resources with value $\epsilon$ and uniquely assign these low value resources to each agent.
In this way, we create a random and non-trivial resource allocation game for any $n$.

To time Algorithm~\ref{algo:min communication}, we generate $20$ random games for each $n$ and report the averaged time for each solution.
In the figure, we can see that the time complexity appears to grow with a convex curve.
This is expected, since Algorithm~\ref{algo:SNE}'s time complexity is $\mathcal{O}(n^2|A_*|)$, noting that in this experiment $|A_*|=11$ is fixed, so we would expect to observe $\mathcal{O}(n^2)$.
We highlight that the largest game considered in this experiment ($n=1000$) has a very large action space $|A|=11^{1000}$ and runs in the millisecond timescale.
This emphasizes that this algorithm is suitable to analyze large domains of games or provide communication graphs in real-time. 


\section{Conclusion}
In this work, we presented a pipeline of algorithms that compute a minimal communication game such that the Price of Anarchy is bounded below by the full communication game.
To support this computation, we provided a characterization of strict Nash equilibria across all communication profiles, and showed the existence of weak best response paths from a Nash equilibrium to a strict Nash equilibrium.
Finally, we highlighted that the proposed minimum communication game can be used to partition the space of communication networks to form high and low information games.
In future work, we plan to study the equilibrium properties of high and low information communication spaces with higher fidelity.
For example, high-information games serve as a sufficient condition for the existence of Nash equilibrium.
However, the conditions in low-information setting such that a Nash equilibrium exists remains an open question.

\bibliography{references}

\begin{thebibliography}{10}

\bibitem{kuntze2012seneka}
H.-B. Kuntze, C.~W. Frey, I.~Tchouchenkov, B.~Staehle, E.~Rome, K.~Pfeiffer, A.~Wenzel, and J.~W{\"o}llenstein, ``Seneka-sensor network with mobile robots for disaster management,'' in {\em 2012 IEEE Conference on Technologies for Homeland Security (HST)}, pp.~406--410, IEEE, 2012.

\bibitem{steinhausler2022detection}
F.~Steinh{\"a}usler and H.~V. Georgiou, ``Detection of victims with uavs during wide area search and rescue operations,'' in {\em 2022 IEEE International Symposium on Safety, Security, and Rescue Robotics (SSRR)}, pp.~14--19, IEEE, 2022.

\bibitem{zhang2022training}
N.~Zhang, F.~Nex, G.~Vosselman, and N.~Kerle, ``Training a disaster victim detection network for uav search and rescue using harmonious composite images,'' {\em Remote Sensing}, vol.~14, no.~13, p.~2977, 2022.

\bibitem{tusnio2022efficiency}
N.~Tu{\'s}nio and W.~Wr{\'o}blewski, ``The efficiency of drones usage for safety and rescue operations in an open area: A case from poland.,'' {\em Sustainability (2071-1050)}, vol.~14, no.~1, 2022.

\bibitem{arnold2018search}
R.~D. Arnold, H.~Yamaguchi, and T.~Tanaka, ``Search and rescue with autonomous flying robots through behavior-based cooperative intelligence,'' {\em Journal of International Humanitarian Action}, vol.~3, no.~1, pp.~1--18, 2018.

\bibitem{eom2019uav}
S.~Eom, H.~Lee, J.~Park, and I.~Lee, ``Uav-aided two-way mobile relaying systems,'' {\em IEEE Communications Letters}, vol.~24, no.~2, pp.~438--442, 2019.

\bibitem{alexis2009coordination}
K.~Alexis, G.~Nikolakopoulos, A.~Tzes, and L.~Dritsas, ``Coordination of helicopter uavs for aerial forest-fire surveillance,'' in {\em Applications of intelligent control to engineering systems}, pp.~169--193, Springer, 2009.

\bibitem{tzoumas2023wildfire}
G.~Tzoumas, L.~Pitonakova, L.~Salinas, C.~Scales, T.~Richardson, and S.~Hauert, ``Wildfire detection in large-scale environments using force-based control for swarms of uavs,'' {\em Swarm Intelligence}, vol.~17, no.~1, pp.~89--115, 2023.

\bibitem{pham2018distributed}
H.~X. Pham, H.~M. La, D.~Feil-Seifer, and M.~C. Deans, ``A distributed control framework of multiple unmanned aerial vehicles for dynamic wildfire tracking,'' {\em IEEE Transactions on Systems, Man, and Cybernetics: Systems}, vol.~50, no.~4, pp.~1537--1548, 2018.

\bibitem{innocente2019self}
M.~S. Innocente and P.~Grasso, ``Self-organising swarms of firefighting drones: Harnessing the power of collective intelligence in decentralised multi-robot systems,'' {\em Journal of computational science}, vol.~34, pp.~80--101, 2019.

\bibitem{alsammak2023nature}
I.~L.~H. Alsammak, M.~A. Mahmoud, S.~S. Gunasekaran, A.~N. Ahmed, and M.~AlKilabi, ``Nature-inspired drone swarming for wildfires suppression considering distributed fire spots and energy consumption,'' {\em Ieee Access}, vol.~11, pp.~50962--50983, 2023.

\bibitem{john2025resource}
J.~John, S.~Velhal, and S.~Sundaram, ``A resource-efficient decentralized sequential planner for spatiotemporal wildfire mitigation,'' {\em IEEE Transactions on Automation Science and Engineering}, vol.~22, pp.~11469--11482, 2025.

\bibitem{comfort2006communication}
L.~K. Comfort and T.~W. Haase, ``Communication, coherence, and collective action: The impact of hurricane katrina on communications infrastructure,'' {\em Public Works management \& policy}, vol.~10, no.~4, pp.~328--343, 2006.

\bibitem{manzini2023wireless}
T.~Manzini, R.~Murphy, D.~Merrick, and J.~Adams, ``Wireless network demands of data products from small uncrewed aerial systems at hurricane ian,'' in {\em 2023 IEEE/RSJ International Conference on Intelligent Robots and Systems (IROS)}, pp.~9941--9946, IEEE, 2023.

\bibitem{mcrae2021utilizing}
J.~N. McRae, B.~M. Nielsen, C.~J. Gay, A.~P. Hunt, and A.~D. Nigh, ``Utilizing drones to restore and maintain radio communication during search and rescue operations,'' {\em Wilderness \& Environmental Medicine}, vol.~32, no.~1, pp.~41--46, 2021.

\bibitem{nayak2020experimental}
S.~Nayak, S.~Yeotikar, E.~Carrillo, E.~Rudnick-Cohen, M.~K.~M. Jaffar, R.~Patel, S.~Azarm, J.~W. Herrmann, H.~Xu, and M.~Otte, ``Experimental comparison of decentralized task allocation algorithms under imperfect communication,'' {\em IEEE Robotics and Automation Letters}, vol.~5, no.~2, pp.~572--579, 2020.

\bibitem{schneider2015auction}
E.~Schneider, E.~I. Sklar, S.~Parsons, and A.~T. {\"O}zgelen, ``Auction-based task allocation for multi-robot teams in dynamic environments,'' in {\em Conference Towards Autonomous Robotic Systems}, pp.~246--257, Springer, 2015.

\bibitem{otte2020auctions}
M.~Otte, M.~J. Kuhlman, and D.~Sofge, ``Auctions for multi-robot task allocation in communication limited environments,'' {\em Autonomous robots}, vol.~44, no.~3, pp.~547--584, 2020.

\bibitem{sujit2012multi}
P.~Sujit and J.~Sousa, ``Multi-uav task allocation with communication faults,'' in {\em 2012 American Control Conference (ACC)}, pp.~3724--3729, IEEE, 2012.

\bibitem{rachmut2023asynchronous}
B.~Rachmut, S.~A. Nelke, and R.~Zivan, ``Asynchronous communication aware multi-agent task allocation.,'' in {\em AAMAS}, pp.~2340--2342, 2023.

\bibitem{su2016coordination}
X.~Su, M.~Zhang, and Q.~Bai, ``Coordination for dynamic weighted task allocation in disaster environments with time, space and communication constraints,'' {\em Journal of Parallel and Distributed Computing}, vol.~97, pp.~47--56, 2016.

\bibitem{rosenthal1973class}
R.~W. Rosenthal, ``A class of games possessing pure-strategy nash equilibria,'' {\em International journal of game theory}, vol.~2, no.~1, pp.~65--67, 1973.

\bibitem{milchtaich1996congestion}
I.~Milchtaich, ``Congestion games with player-specific payoff functions,'' {\em Games and economic behavior}, vol.~13, no.~1, pp.~111--124, 1996.

\bibitem{ieong2005fast}
S.~Ieong, R.~McGrew, E.~Nudelman, Y.~Shoham, and Q.~Sun, ``Fast and compact: A simple class of congestion games,'' in {\em Proceedings of the 20th national conference on Artificial intelligence-Volume 2}, pp.~489--494, 2005.

\bibitem{ashlagi2006resource}
I.~Ashlagi, D.~Monderer, and M.~Tennenholtz, ``Resource selection games with unknown number of players,'' in {\em Proceedings of the fifth international joint conference on Autonomous agents and multiagent systems}, pp.~819--825, 2006.

\bibitem{ashlagi2007learning}
I.~Ashlagi, D.~Monderer, and M.~Tennenholtz, ``Learning equilibrium in resource selection games,'' in {\em PROCEEDINGS OF THE NATIONAL CONFERENCE ON ARTIFICIAL INTELLIGENCE}, vol.~22, p.~18, Menlo Park, CA; Cambridge, MA; London; AAAI Press; MIT Press; 1999, 2007.

\bibitem{gharesifard2017distributed}
B.~Gharesifard and S.~L. Smith, ``Distributed submodular maximization with limited information,'' {\em IEEE transactions on control of network systems}, vol.~5, no.~4, pp.~1635--1645, 2017.

\bibitem{grimsman2018impact}
D.~Grimsman, M.~S. Ali, J.~P. Hespanha, and J.~R. Marden, ``The impact of information in distributed submodular maximization,'' {\em IEEE Transactions on Control of Network Systems}, vol.~6, no.~4, pp.~1334--1343, 2018.

\bibitem{sun2020distributed}
H.~Sun, D.~Grimsman, and J.~R. Marden, ``Distributed submodular maximization with parallel execution,'' in {\em 2020 American Control Conference (ACC)}, pp.~1477--1482, IEEE, 2020.

\bibitem{konda2022execution}
R.~Konda, D.~Grimsman, and J.~R. Marden, ``Execution order matters in greedy algorithms with limited information,'' in {\em 2022 American Control Conference (ACC)}, pp.~1305--1310, IEEE, 2022.

\bibitem{grimsman2022impact}
D.~Grimsman, M.~R. Kirchner, J.~P. Hespanha, and J.~R. Marden, ``The impact of measurement passing in sensor network measurement selection,'' {\em IEEE Transactions on Control of Network Systems}, vol.~10, no.~1, pp.~112--123, 2022.

\bibitem{koutsoupias1999worst}
E.~Koutsoupias and C.~Papadimitriou, ``Worst-case equilibria,'' in {\em Annual symposium on theoretical aspects of computer science}, pp.~404--413, Springer, 1999.

\bibitem{grimsman2020cost}
D.~Grimsman, J.~H. Seaton, J.~R. Marden, and P.~N. Brown, ``The cost of denied observation in multiagent submodular optimization,'' in {\em 2020 59th IEEE Conference on Decision and Control (CDC)}, pp.~1666--1671, IEEE, 2020.

\bibitem{seaton2022all}
J.~H. Seaton and P.~N. Brown, ``All stable equilibria have improved performance guarantees in submodular maximization with communication-denied agents,'' {\em IEEE Control Systems Letters}, vol.~6, pp.~2491--2496, 2022.

\bibitem{grimsman2022valid}
D.~Grimsman, P.~N. Brown, and J.~R. Marden, ``Valid utility games with information sharing constraints,'' in {\em 2022 IEEE 61st Conference on Decision and Control (CDC)}, pp.~5739--5744, IEEE, 2022.

\bibitem{brown2023robust}
P.~N. Brown, J.~H. Seaton, and J.~R. Marden, ``Robust networked multiagent optimization: designing agents to repair their own utility functions,'' {\em Dynamic Games and Applications}, vol.~13, no.~1, pp.~187--207, 2023.

\bibitem{singh2025optimal}
V.~Singh, W.~Wesley, and P.~N. Brown, ``Optimal utility design with arbitrary information networks,'' in {\em 2025 American Control Conference (ACC)}, pp.~2895--2900, IEEE, 2025.

\bibitem{singh2025worst}
V.~Singh and P.~N. Brown, ``Worst-case equilibria in networked resource allocation games rest on a knife-edge,'' in {\em 2025 IEEE 64th Conference on Decision and Control (CDC)}, pp.~4409--4414, IEEE, 2025.

\bibitem{monderer1996potential}
D.~Monderer and L.~S. Shapley, ``Potential games,'' {\em Games and economic behavior}, vol.~14, no.~1, pp.~124--143, 1996.

\bibitem{marden2012revisiting}
J.~R. Marden and J.~S. Shamma, ``Revisiting log-linear learning: Asynchrony, completeness and payoff-based implementation,'' {\em Games and Economic Behavior}, vol.~75, no.~2, pp.~788--808, 2012.

\bibitem{kleer2017potential}
P.~Kleer and G.~Sch{\"a}fer, ``Potential function minimizers of combinatorial congestion games: Efficiency and computation,'' in {\em Proceedings of the 2017 ACM Conference on Economics and Computation}, pp.~223--240, 2017.

\bibitem{Bilo2025Minimizing}
V.~Bil{\`o}, A.~Fanelli, L.~Gourv{\`e}s, C.~Tsoufis, and C.~Vinci, ``Minimizing rosenthal's potential in monotone congestion games,'' in {\em Proceedings of the 24th International Conference on Autonomous Agents and Multiagent Systems}, AAMAS '25, (Richland, SC), p.~343–351, International Foundation for Autonomous Agents and Multiagent Systems, 2025.

\end{thebibliography}
\bibliographystyle{ieeetr}

\appendices
\section{Proofs}

\subsection*{Proof of Theorem~\ref{thm:sne exists}}
\begin{proof}
    Let full communication game $\bar{g}\in G$ satisfy assumptions A1 and A2.
    Assume by contradiction SNE$(\bar{g})=\emptyset$ is the empty set.
    Let $A^I=\{a\in A: |a|_r\leq1, \forall r\in R\}$ be the set of informed actions and by Assumption A1, $A^I$ is non-empty.
    Consider an informed action profile $a\in A^I$, and by assumption $a$ cannot be a SNE.
    By definition there exists some agent $i\in N$ with action $a'_i\in A_i$ such that $U_i(a'_i,a_{-i})\geq U(a_i,a_{-i})$.
    However, it cannot be $U_i(a'_i,a_{-i})= U(a_i,a_{-i})$ because $a$ is informed and by assumption A2.
    Particularly, $U_i(a)=v_{a_i}\neq v_{a_i'}=U(a'_i,a_{-i})$.
    Therefore the inequality must be $U_i(a'_i,a_{-i})> U(a_i,a_{-i})$ and the new action profile $a'=(a'_i,a_{-i})\in A^I$ as $U_i(a'_i,a_{-i})>0$ implies $|a'|_{a'_i}=1$.
    Now consider the sum 
    $$f(a)=\sum_{i\in N}U_i(a)$$
    and observe that $U_j(a)=U_j(a')$ for all $j\neq i$ for $a,a'$ as constructed above.
    Therefore, we can construct an $a'$ such that $f(a')>f(a)$ for any $a\in A^I$.
    Consider the problem of finding the max of function $f$ over $A^I$, given by $\max_{a\in A^I}f(a)$.
    However, this maximum does not exist because for any proposed maximum $a\in A^I$ there exists some $a'\in A^I$ such that $f(a')>f(a)$.
    But $\{f(a): a\in A^I\}$ is a finite set of real numbers and must have a maximum.
    This is a contradiction, and therefore SNE$(\bar{g})$ must be a non-empty set.
\end{proof}

\subsection*{Proof of Theorem~\ref{thm:informed ne}}
\begin{proof}
    Let full communication game $\bar{g}\in \bar{G}$ satisfy Assumptions A1 and A2.
    For clarity, we refer to the utility functions associated with $\bar{g}$ as $\bar{U}_i$ throughout this proof.
    We begin with the forward direction, assuming that $a^*\in \textrm{SNE}(\bar{g})$ and we show that $a^*\in\cup_{g\in G(\bar{g})} \textrm{NE}(g)$ and $|a^*|_r\leq 1$ for all $r\in R$.
    First, we observe that $a^*\in \textrm{SNE}(\bar{g})$ exists by Theorem~\ref{thm:sne exists} (this dependency requires Assumptions A1 and A2) as $\textrm{SNE}(\bar{g})$ is non-empty.
    Next, $a^*\in\cup_{g\in G(\bar{g})} \textrm{NE}(g)$ is trivial by assumption as $\bar{g}\in G(\bar{g})$.
    All that remains is to show that $a^*$ is an informed action, that is, $|a|_r\leq 1$ for all $r\in R$.
    To show this, suppose by contradiction there exists some $i\in N$ such that $|a^*|_{a_i^*}\geq 2$.
    Then $\bar{U}_i(a^*)=W(a^*_i,a^*_{-i})-W(a^*_{-i})=0$ as $|a^*_{-i}|\geq 1$ by marginal contribution.
    In the case $|A_i|\geq 2$, then there exists $a_i\neq a^*_i$ such that $U(a_i,a^*_{-i})\geq U(a^*)$.
    If $|A_i|=1$, then by Assumption A1 the agent $j$ with $a^*_j=a^*_i$ must have $|A_i|\geq2$, and the above contradiction can be applied to $j$.
    Therefore, by contradiction if $a^*\in \textrm{SNE}$ then it must be $|a^*|_r\leq 1$ for all $r\in R$.
    This concludes the forward direction.

    Moving to the backward direction, we assume that $a^*\in\cup_{g\in G(\bar{g})} \textrm{NE}(g)$ and $|a^*|_r\leq 1$ for all $r\in R$ and show $a^*\in \textrm{SNE}(\bar{g})$.
    The proof proceeds in two cases on $g\in G(\bar{g})$ such that $a^*\in \textrm{NE}(g)$.
    We begin with the first case and assume that $g=\bar{g}$ has full communication.
    By definition of Nash \eqref{eq:NE} we have that $\bar{U}_i(a^*)\geq \bar{U}_i(a'_i,a^*)$ for any action $a'_i$ for all agents $i\in N$. 
    It suffices to show that this inequality must be strict, which is true by Assumption A2 and $a^*$ being informed.
    Informed gives that $U_i(a^*)>0$ (meaning no other agents are cover the same resource) and Assumption A2 gives that $U_i(a'_i,a^*_{-i})\neq U_i(a^*_i,a^*_{-i})$ as any resources $a'_i=r'$ and $a^*_i=r$ must have $v_r\neq v_{r'}$.

    In the second case, we consider $a^*\in \textrm{NE}(g)$ for some $g\in G(\bar{g})\setminus\{\bar{g}\}$ with communication denial.
    We denote the utility function of game $g$ as $U_i$, which we highlight as distinct from the one from game $\bar{g}$ which is $\bar{U}_i$. 
    We first observe that $\bar{U}(a^*)=U(a^*)=v_r$ where $r=a^*_i$.
    This follows because $a^*$ is informed, that is because $i$ is the only agent covering $r$, $\mathcal{N}_i$ does not impact the utility.
    However, if we consider arbitrary unilateral deviation $a'_i\in A_i\setminus\{a^*_i\}$ from $a^*$ by any agent $i$, then we have $\bar{U}(a'_i,a^*_{-i})\leq U(a'_i,a^*)$.
    This is because either  $|a_i,a^*_{-i}|_{r'}\geq 2$, or $|a_i,a^*_{-i}|_{r'}=1$ (where $a'_i=r'$).
    In the first case $\bar{U}(a^*_i,a^*_{-i})=0$ and in the second case $\bar{U}(a^*_i,a^*_{-i})=v'_r$.
    Putting these two observations together with the definition of $a^*$ being a Nash equilibrium, we get the following inequality chain:
    $\bar{U}(a^*_i,a^*_{-i})=U(a^*_i,a^*)\geq U(a'_i,a^*)\geq \bar{U}(a'_i,a^*)$.
    This gives that $a^*\in \textrm{NE}(\bar{g})$.
    Then, the argument in the first case $g=\bar{g}$ can be applied as is resulting in $a^*\in \textrm{SNE}(\bar{g})$ as desired.
    This concludes the backward direction of the proof.
\end{proof}
\subsection*{Proof of Theorem~\ref{thm:gmin}}
\begin{proof}
    Let $\gmin\in G(\bar{g})$ be a minimal communication game with respect to some $a^*\in \textrm{SNE}(\bar{g})$ for a full information game $\bar{g}\in \bar{G}$ that satisfies Assumptions A1 and A2.
    Let $\mathcal{N}_i$ denote the communication structure for game $\gmin$.
    We begin by verifying the first statement that $\textrm{NE}(g)=\{a^*\}$, noting that the the second claim follows immediately from the first.

    We prove the first claim in two parts, first showing $a^*\in \textrm{NE}(g)$ and then that it is the unique Nash equilibrium.
    Consider some agent $i$'s decision to select $a^*_i$ in contrast some distinct $a_i\in A_i$, given by $U(a^*_i,a^*_{-i})-U(a_i,a^*_{-i})=v_{a^*_i}-U(a_i,a^*_{-i})$, where the equality follows by Theorem~\ref{thm:informed ne} and therefore $a^*$ is informed.
    We verify this difference is non-negative for all $i\in N$ and $a_i\in A_i$, giving $a^*$ must be a Nash equilibrium
    We consider two cases: $a_i=a^*_j$ for some $j$ and $a_i\neq a^*_j$ for all $j$.
    Beginning with $a_i=a^*_j$, this implies that $a^*_j\in A_i$ fulfilling the first condition of the if of Algorithm~\ref{algo:min communication}.
    Then, if $v_{a^*_j}>v_{a^*_i}$ the Algorithm ensures that $j\in \mathcal {N}_i$ giving $U(a_i,a^*_{-i})=0$ ensuring the difference is positive.
    Otherwise, if $v_{a^*_j}<v_{a^*_i}$ then $U(a^*_i,a^*_{-i})-U(a_i,a^*_{-i})=v_{a^*_i}-v_{a^*_j}>0$.
    Finally, in the case that $a_i\neq a^*_j$ for all $j\in N$, the following statement holds:
    \begin{equation}
        U(a^*_i,a^*_{-i})=\bar{U}(a^*_i,a^*_{-i})>\bar{U}(a_i,a^*_{-i})=U(a_i,a^*_{-i}).
    \end{equation}
    The first equality hold because Theorem~\ref{thm:informed ne} gives that $a^*$ must be informed, the inequality holds by definition of $a^*\in \textrm{SNE}(\bar{g})$, and the third equality holds by assumption that $a_i\neq a^*_j$ for all $j\in N$.
    Therefore, $U(a^*_i,a^*_{-i})>U(a_i,a^*_{-i})$ for any $i\in N$ and $a_i\in A_i\setminus\{a^*_i\}$ which gives that $a^*\in \textrm{NE}(g)$.

    Now that we have established that $a^*\in \textrm{NE}(g)$, we further show that $\textrm{NE}(g)=\{a^*\}$
    To do this, sort the agents into the sequence $(i^1,i^2,i^3,...,i^n)$ such that $v_{a^*_{i^t}}>v_{a^*_{i^{t+1}}}$.
    We use strong induction on this sequence to show that for any $a\in \textrm{NE}(g)$, it must be that $a_{i^t}=a^*_{i^t}$ for all $t\in \{1,\dots,n\}$. 
    First, we verify the base case $a_{i^1}=a^*_{i^1}$ holds.
    Because we have that $v_{a^*_{i^t}}=\max_{r\in R} v_r$, by Algorithm~\ref{algo:min communication} it must by the if condition in $\mathcal{N}_{i^1}=\emptyset$.
    Therefore, $U_i(a^*_{i^1},a_{-i^1})=\max_{r\in R} v_r$ regardless of $a_{-i^1}$, meaning $a_{i^1}=a^*_{i^1}$ because $a$ is a strict Nash equilibrium.
    Now, we move to the inductive step and assume that $a_{i^{k}}=a^*_{i^{k}}$ holds for all $k< t$ and we seek to show that $a_{i^{t}}=a^*_{i^{t}}$.
    Now we partition the agents into two sets $N_{<t}=\{i^1,\dots i^k\}$ and $N_{>t}=\{i^{t+1},\dots, i^n\}$.
    For the first set, taking $i^k\in N_{<t}$, we have $a_{i^k}=a^*_{i^k}$, and by the if statement of Algorithm~\ref{algo:min communication} that if $a_{i^k}\in A_{i^t}$ then $i^k\in \mathcal{N}_i$.
    Now, consider $i^k\in N_{>t}$, it must be $i^k\notin \mathcal{N}_{i^t}$ because $v_{a^*_{i^k}}<v_{a^*_{i^t}}$.
    Putting these two observations together, anytime  $i^k\in \mathcal{N}_{i^t}$, it must be that $a_{i^k}=a^*_{i^k}$.
    Therefore, we get the following utility equality
    \begin{equation}
        U_{i^{t}}(a'_i,a^*_{-i^{t}})=U_{i^{t}}(a'_i,a_{-i^{t}})
    \end{equation}
    for all $a'_i\in A_i$.
    By definition that $a^*_{i^t}\in \textrm{SNE}(g)$, it must be that $a_{i^t}=a^*_{i^t}$, completing the induction.
    Therefore, anytime $a\in \textrm{NE}(g)$, it must be that $a=a^*$, meaning that $a^*$ is the unique element of $\textrm{NE}(g)$ as desired.
    Further, because 
    \begin{equation}
    \{a^*\}=\textrm{SNE}(g)=\textrm{NE}(g)\subseteq\textrm{SNE}(\bar{g})\subseteq\textrm{NE}(\bar{g})   
    \end{equation}
    it must be that $\textrm{sPoA}(g)\geq \textrm{sPoA}(\bar{g})$ and $\textrm{PoA}(g)\geq \textrm{PoA}(\bar{g})$ completing the second claim.

    We now verify the third claim,
    letting $g'\in G(\bar{g})$ have $\mathcal{N}_i\nsubseteq \mathcal{N}'_i$ for some $i\in N$ and by contradiction assume $a^*\in \textrm{SNE}(g')$.
    By definition of $g'$ must exist some $j\in N$ such that $j\in \mathcal{N}_i$ but $j\notin \mathcal{N}'_i$.
    Now consider agent $i$'s utility between selecting $a^*_j$ and $a^*_i$:
    \begin{equation}\label{}
        U'_i(a^*_j,a^*_{-i})-U'_i(a^*_i,a^*_{-i})
    = v_{a^*_j}-v_{a^*_i} > 0.
    \end{equation}
    The first equality follows because $U'_i(a^*_j,a^*_{-i})=v_{a^*_j}$ and $U'_i(a^*_i,a^*_{-i})=v_{a^*_i}$.
    The first follows because $j\notin \mathcal{N}'_i$ and $a^*_j$ is uniquely covered by agent $j$ and the second because $a^*_i$ is uniquely covered by agent $i$ (both unique coverings follow from Theorem~\ref{thm:informed ne}). 
    Finally, the inequality holds  because $j\in \mathcal{N}_i$ it must have been $v_{a^*_j}>v_{a^*_i}$ by the if statement in Algorithm~\ref{algo:min communication}.
    Therefore, any such game $g'$ has $a^*\notin \textrm{SNE}(g')$ and therefore $g$ must the minimal communication game.

    We now verify the fourth claim, letting $g''\in G(\bar{g})$ have $\mathcal{N}_i\subseteq \mathcal{N}''_i$ for all $i\in N$.
    The following sequence expression holds for any $a_i\in A$ such that $a_i\neq a^*_i$:
    \begin{equation} \label{eq:proof:util ineq}
    U''_i(a^*_{i},a^*_{-i})=U_i(a^*_{i},a^*_{-i})>U_i(a_i,a^*_{-i})\geq U''_i(a_i,a^*_{-i}).
    \end{equation}
    The first equality follows because by Theorem~\ref{thm:informed ne} gives that  $a^*$ is informed, meaning that regardless of $\mathcal{N}_i,\mathcal{N}''_i$, the utility of $a^*_i$ does not change.
    The first inequality follows by definition of $a^*\in \textrm{SNE}(g)$, and the second inequality follows from the taking four cases on $a_i:$
    \begin{equation}
    U_i(a_i,a^*_{-i})-U_i''(a_i,a^*_{-i})=\begin{cases}
        0 & \nexists j\in N \textrm{ s.t. }  a_i=a^*_j\\
        0 & a_i=a^*_j, j\in \mathcal{N}_i \\
        v_{a_i} & a_i=a^*_j, j\notin \mathcal{N}_i ,j\in \mathcal{N}''_i \\
        0 & a_i=a^*_j, j\notin \mathcal{N}_i, j\notin \mathcal{N}''_i \\
    \end{cases}
    \end{equation}
    noting that $\mathcal{N}_i\subseteq \mathcal{N}''_i$ gives that we do not need to consider the case where $a_i=a^*_j$ and $j\in\mathcal{N}_i$ but $j\notin\mathcal{N}_i$.
    Because all cases are non-negative, this verifies the second inequality of \eqref{eq:proof:util ineq}, giving that $U''_i(a^*_{i},a^*_{-i})>U''_i(a_i,a^*_{-i})$ for any $a_i\in A_i$ such that $a_i\neq a^*_i$.
    This is the definition of strict Nash equilibrium and therefore $a^*\in \textrm{SNE}(g'')$ as desired and completing the fourth claim.
\end{proof}

\subsection*{Proof of Lemma~\ref{thm:monotinicity}}
\begin{proof}
    Let $a^1,a^2\in A$ be action profiles resulting from Algorithm~\ref{algo:NE} on $\bar{g}\in \bar{G}$.
    Throughout, we use $a^{t,i}_{-i}$ to refer to the action profile that agent $i$ observed when they made their update in iteration $t$.
    Assume by contradiction that $R(a^1)\nsubseteq R(a^2)$, meaning there exists some resource $r\in R$ such that $r\in R(a^1)$ but $r\notin R(a^2)$.
    Particularly, there must exist an agent $i$ has $a^1_i=r$ but $a^2_i=r'$ such that $r\notin R(a^{2,i}_{-i})$.
    By definition of the algorithm, it must be that $U_i(a_i^2,a^{2,i}_{-i})>U_i(a_i^1,a^{2,i}_{-i})$, but $U_i(a_i^2,a^{1,i}_{-i})\leq U_i(a_i^1,a^{1,i}_{-i})$.
    First, observe that $U_i(a_i^1,a^{2,i}_{-i})=v_r$ because $r\notin R(a^{2,i}_{-i})$, and for the strict inequality to hold it must be $r'\notin R(a^{2,i}_{-i})$ and therefore $U_i(a_i^2,a^{2,i}_{-i})=v_{r'}$, giving $v_{r}<v_{r'}$ (Inequality is strict due to Assumption A2).
    Second, it must have been that $r'\in R(a^{1,i}_{-i})$ but $r'\notin R(a^{2,i}_{-i})$ (or else $a^1_i=r'$), meaning there exists an agent $j$ such that $a^1_j=r'$ and $a^2_j\neq r'$, and $r'\notin R(a^{2,j}_{-j})$.
    Agent $j$ now has the exact same conditions placed on them as agent $i$ originally, and the argument can be repeated verbatim. 
    Particularly, that there must exist resource $a^2_j=r''$ such that $U_j(a_j^2,a^{2,j}_{-j})>U_j(a_j^1,a^{2,j}_{-j})$, but $U_j(a_j^2,a^{1,j}_{-j})\leq U_j(a_j^1,a^{1,j}_{-j})$.
    The resulting inequality sequence $v_{r}<v_{r'}<v_{r''}<\dots$ can be repeated infinitely, which is a contradiction as $R$ is a finite set.
    Therefore, by contradiction $R(a^1)\subseteq R(a^2)$.
\end{proof}

\subsection*{Proof of Proposition~\ref{thm:find Nash}}
\begin{proof}
    Let $\bar{g}\in \bar{G}$ be a full information game that satisfies A2 and $a^1,a^2\in A$ refer to the action profiles from running Algorithm~\ref{algo:NE} on $\bar{g}$.
    Throughout, we use $a^{t,i}_{-i}$ to refer to the action profile that agent $i$ observed when they made their update in iteration $t$.
    Suppose by contradiction there exists $i\in N$ and action $a_i\in A_i$ such that $U_i(a_i,a^2_{-i})> U_i(a^2)$.
    Starting with case  $U_i(a^2)>0$, 
    we observe that $U_i(a_i,a^2)>U(a^2)>0$ gives that $U_i(a_i,a^2)=v_{a_i}>v_{a^2_i}=U_i(a^2)$.
    Second, $U_i(a_i,a^2)>0$ implies $a_i\notin R(a^2)$.
    By Lemma~\ref{thm:monotinicity} it must have been that $a_i\notin R(a^1)$, which implies that $U(a_i,a^{2,i}_{-i})=v_{a_i}$, but Algorithm~\ref{algo:NE} gave $\{a^2_i\}=\argmax_{a'_i \in \tilde{A}_i} U_i(a_i,a^{2,i}_{-i})$ (uniqueness comes from Assumption A2).
    This implies $v_{a^2_i}>v_{a_i}$ which is a contradiction.
    Therefore, it cannot be that $U_i(a^2)>0$.
    
    We consider the second case $U_i(a^2)=0$, once again we have $a_i\notin R(a^2)$ as it must be $U_i(a_i,a^2_{-i})>U_i(a^2_i,a^2_{-i})$.
    By Lemma~\ref{thm:monotinicity} (which gives requires Assumption A2) it must be that $a_i\notin R(a^1)$ which gives that $U_i(a_i,a^{2,i}_{-i})=v_{a_i}$ and therefore $U_i(a^2_i,a^{2,i}_{-i})=v_{a^2_i}>v_{a_i}$.
    For it to be true that $U_i(a_i,a^2_{-i})>U_i(a^2_i,a^2_{-i})$ it must be that some agent $j$ selected $a^2_i\in A_j$ after agent's $i$'s turn in iteration 2.
    However, because such an update occurs after agent $i$ selected $a_i$, it must be that $U_j(a^2_i,a^{2,j}_{-j})=0$ meaning it could not have been a strict best response as required by the algorithm.
    Therefore, such a selection is not possible and assumption $U_i(a^2)=0$ is contradicted by $U_i(a^2_i,a^{2}_{-i})=v_{a^2_i}$, which follows from the above argument.
\end{proof}

\subsection*{Proof of Theorem~\ref{thm:SWBRP}}
\begin{proof}
    Let $\bar{g}\in \bar{G}$ satisfy Assumption A1 and $a^0\in \textrm{NE}(\bar{g})$ have $|R(a^0)|<n$.
    By A1, it must be that there exists $a^*\in A$ such that $a^*$ is informed. 
    We now consider a  sequential procedure such that agents update their action from $a^0$ to $a^*$.
    The update process proceeds in rounds, where $\rho(t)\subseteq N$ agents update their action at time $t$. 
    Particularly, each agent $i\in \rho(t)$ updates their action from $a_i$ to $a^*_i$.
    Then, in the next round, the agents that update are given by $\rho(t+1)=\{i\in N : a^0_i=a^*_j\textrm{ for }j\in \rho(t)\}$.
    The process begins with $\rho(1)=\{s\}$ for any $s\in S(a^0)$ ($|R(a^0)|<n$ ensures $S(a^0)$ is non-empty) and terminates when at timestep $T$ when $d\in \rho(T)$ for any $d\in D(a^0)$.
    The weak best response sequence can then be created by taking the sequence of agents that led to $d \in \rho(T)$, particularly it must be some agent $i$ selected $a_i=a^0_d$. Further this agent is unique by definition of $a^*$.
    This process can be repeated iteratively, until a sequence of agents $(i^1,i^2,\dots i^T)$ is recovered.
    Then a weak best response sequence can be constructed by taking $a^t=(a^*_{i^t},a^{t-1}_{-i^t})$, creating sequence $(a^0,a^1,\dots a^{T-1})$ up to the second to last agent. For the last update by agent, $i^T=d$ a max is taken
    $$a^T_d\in \argmax_{a_d\in A_d} U(a_d,a^{T-1}_{-d}).$$
    This creates a weak best response sequence, because only one agent updates their action at a time, by definition of the $i\in \rho(t)$ there must have been $j\in \rho(t-1)$ such that $a^t_j=a^0_i$, and each agent can only update once because $a^*$ is informed.
    The definition of $S(a^0),C(a^0)$ ensure that all updates in $a^0\dots a^{T-1}$ were weak best responses, and the final update in $a^T$ is by definition a best response.
    
    To show this process always terminates within $T\leq n$ (particularly verifying $T$ is finite), we show the property that $|\rho(t+1)|\geq |\rho(t)|$.
    This property follows because for any $i\in \rho(t)$ for $t<T$ it must be that $|a^0_{-i}|_{a_i^*}\geq 1$  by definition that $i\in S(a^0)\cup C(a^0)$.
    This implies that when $i$ selects $a^*_i$, at least one agent will be placed in $\rho(t+1)$ as a result, which implies $|\rho(t+1)|\geq |\rho(t)|$.
    Then, to see $T\leq n$, we use the fact that each update can only update once and $|\rho(t)|\geq |\rho(1)|= 1$, meaning it will take at most $n$ steps. 
    
    We now show $R(a^0)\subset R(a^T)$ and $a^T\in \textrm{NE}(\bar{g})$.
    Beginning with we observe anytime an agent $i$ updated their action from $a_i$ it must have been that agent $j$ had $a^*_j=a_i$, therefore if $r\in R(a^t)$ it must be $r\in R(a^{t+1})$.
    The subset $R(a^1)\subset R(a^T)$ is strict because on the last step it is guaranteed that $a^T_d\notin R(a^{T-1})$, therefore $R(a^{T})=R(a^{T-1})\cup\{a^T_d\}$ giving the desired subset property.
    Finally, we show $a^T$ is a Nash equilibrium.
    By \eqref{eq:NE}, we must show that $U_i(a'_i,a^T_{-i})\geq U_i(a^T_i,a^T_{-i})$ for any action $a'_i\in A_i$ for every agent $i$.
    To show this, suppose that agent $i$ updates there action at time $t<T$, and consider agent $i$'s utility $U_i(a_i^t,a_{-i}^{t})$.
    We make a few observations about this function: first $a_i^t=a^{T}_i$ as agent $i$ will never update their action.
    Second, $U_i(a_i',a_{-i}^{t})=U_i(a_i',a_{-i}^{T})$ for all actions $a'_i\in A_i\setminus\{a^T_i\}$, follows from $R(a^{t})\subset R(a^T)$.
    Because $i\in S(a)\cup C(a)$ it must be $U_i(a_i',a^T_{-i})=U(a'_i,a^t)=0$ for all $a'_i\in A_i\setminus\{a^T_i\}$
    Therefore, because  $U_i(a^T_i,a^T_{-i})\geq 0$ it must be $a^T_i$'s Nash inequality is satisfied.
    Finally, when $t=T$ and agent $d$ updates their action, it is defined a best response so the Nash inequality \eqref{eq:NE} holds for agent $d$, which concludes that $a^T$ must be a Nash equilibrium.
\end{proof}

\vspace{-1cm}
\begin{IEEEbiography}
[\raggedbottom{\includegraphics[width=1in,height=1.25in,clip,keepaspectratio]{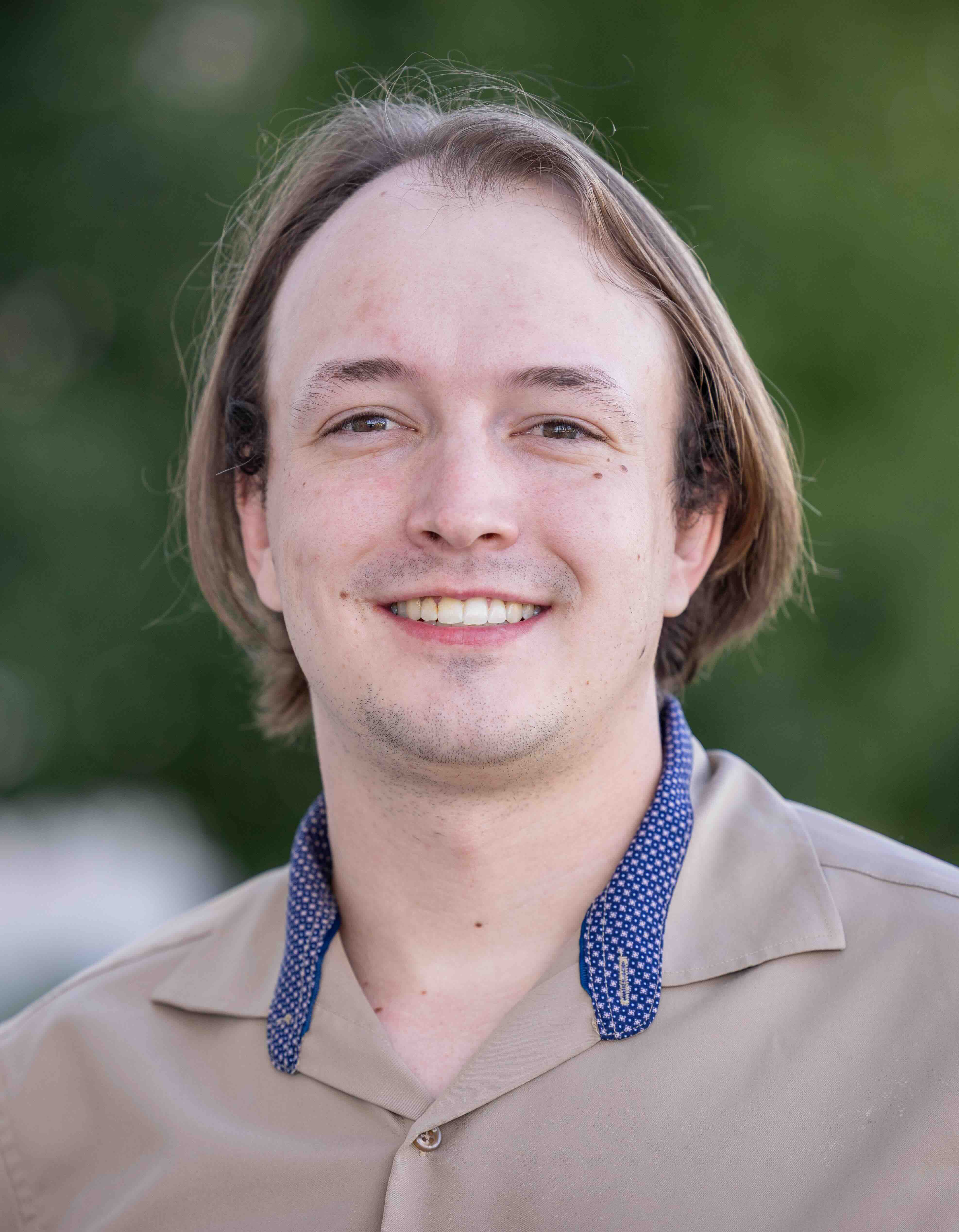}}]{Brandon Collins}(S'19-M'25) is postdoctoral fellow at the Thayer School of Engineering at Dartmouth College.
He received his B.S. in Computer Science and his Ph.D. in 2025, both from the University of Colorado Colorado Springs.
His Ph.D. was advised by Dr. Philip N. Brown and Dr. Shouhuai Xu, and received the outstanding Ph.D. student award.
He is currently participating in the ASEE eFellows program under the mentorship of Dr. Bryce Ferguson.
His research interests include multi-agent systems and formal decision science for cybersecurity.

\end{IEEEbiography}
\vspace{-1cm}

\begin{IEEEbiography}[{\includegraphics[width=1in,height=1.25in,clip,keepaspectratio]{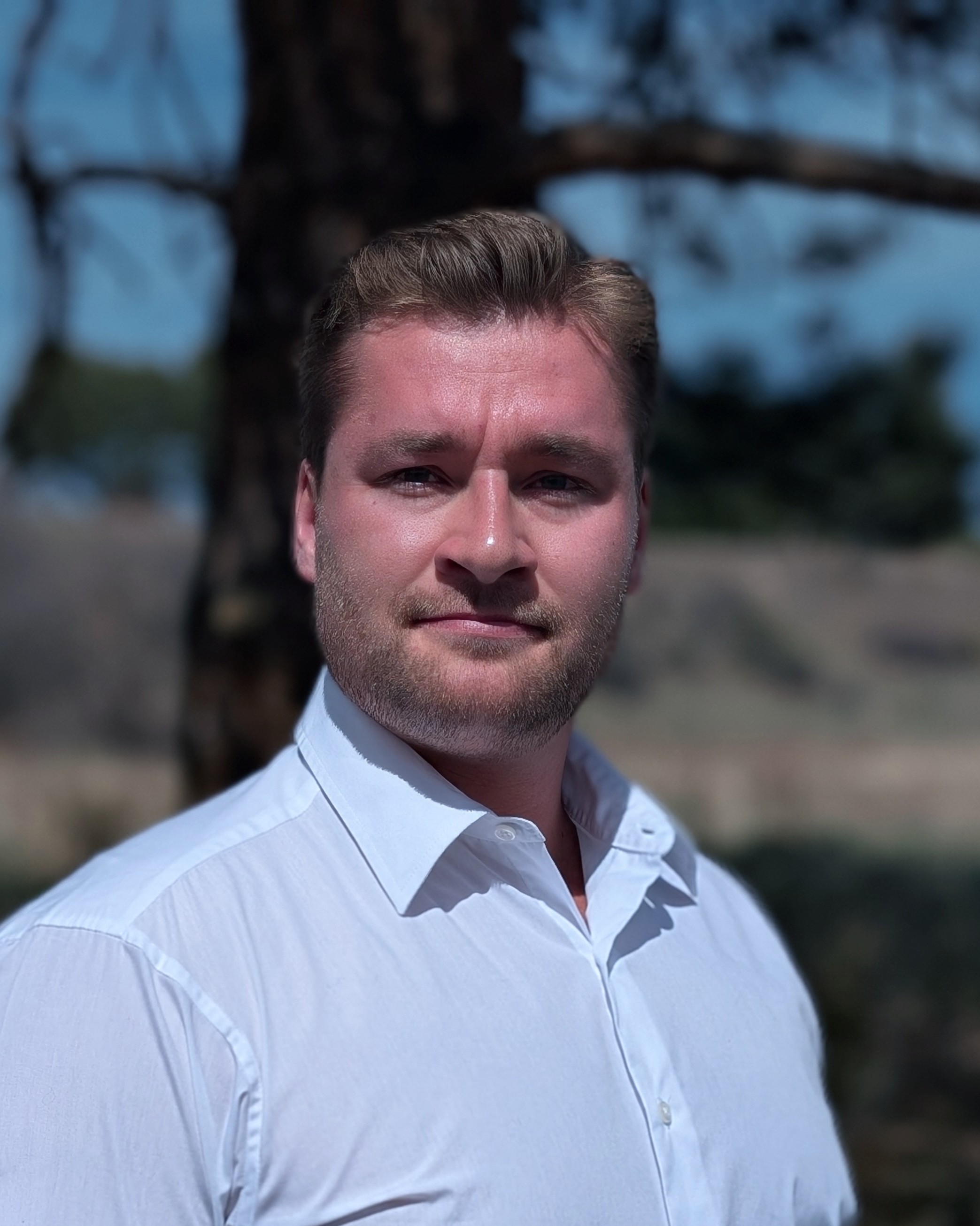}}]{Colton Hill}
is a Postdoctoral Research Associate at the University of Colorado Colorado Springs, where he received the Bachelor of Science in Mathematics in 2020 and the Ph.D. in Computer Science in 2026. His research interests include game theory and multi-agent systems, with an emphasis on efficiency, stability, incentive design, and coupled behavioral–environmental dynamics.

\end{IEEEbiography}

\begin{IEEEbiography}[{\includegraphics[width=1in,height=1.25in,clip,keepaspectratio]{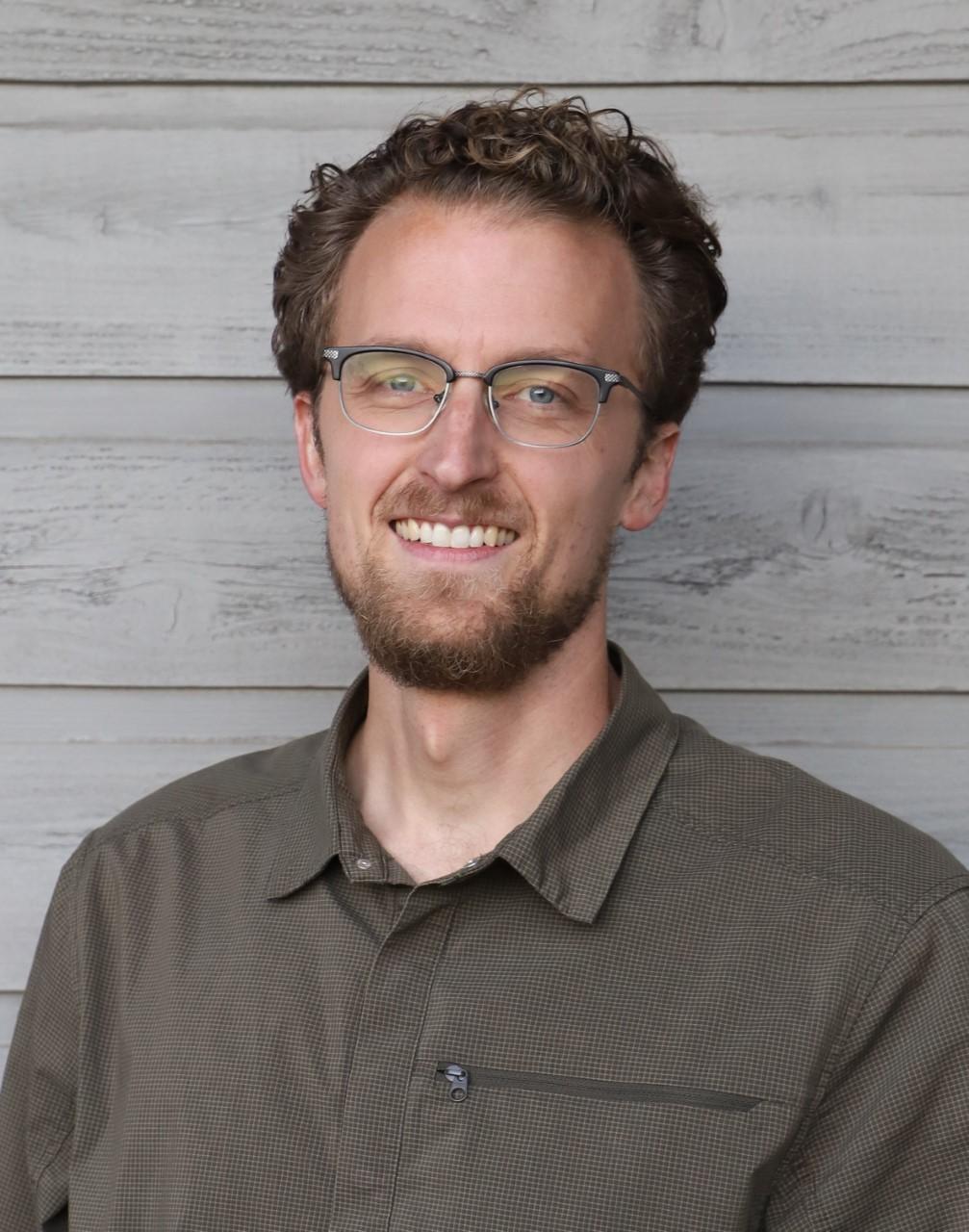}}]{Philip N. Brown} (Member) is Associate Professor in the Department of Computer Science at the University of Colorado Colorado Springs and was a Visiting Professor at Politecnico di Torino in 2026. Philip received the Bachelor of Science in Electrical Engineering in 2007 from Georgia Tech, the Master of Science in Electrical Engineering in 2015 from the University of Colorado at Boulder, and the PhD in Electrical and Computer Engineering from the University of California, Santa Barbara under the supervision of Jason R. Marden. He has received early-career awards from the Air Force Office of Scientific Research (YIP, 2023), National Science Foundation (CAREER, 2025), and the Army Research Office (ECP, 2025). Philip is interested in complex strategic interactions in networked systems of agents, and studies these with a combination of game theory and machine learning. Key application areas include smart transportation systems, human-machine teaming, and space domain awareness.
\end{IEEEbiography}

\end{document}